\documentclass[11pt]{article}

\usepackage{amssymb,amsfonts,amsmath,amsthm,amscd,mathrsfs,dsfont,bbm}
\usepackage{graphicx,float,epsfig,psfrag, color, esdiff}
\usepackage{wrapfig}
\usepackage{subfig}

\DeclareMathAlphabet{\mathpzc}{OT1}{pzc}{m}{it}

\usepackage{fullpage}

\newtheorem{propo}{Proposition}[section]

\newtheorem{lemma}[propo]{Lemma}

\newtheorem{corollary}[propo]{Corollary}
\newtheorem{thm}[propo]{Theorem}

\theoremstyle{definition}
\newtheorem{remark}[propo]{Remark}

\def\<{\langle}
\def\>{\rangle}

\def\eps{{\varepsilon}}
\def\beps{\bar{\varepsilon}}

\def\cY{{\mathcal{Y}}}

\def\sT{{\sf T}}

\def\cH{{\cal H}}

\def\sign{{\operatorname{\rm{sgn}}}}

\def\P{{\mathbb P}}
\def\prob{{\mathbb P}}
\def\integers{{\mathbb Z}}

\def\E{{\mathbb E}} %expectation

\def\one{\mathbf{1}}

\def\reals{\mathbb{R}}

\def\normal{{\sf N}}
\def\Ber{{\sf Ber}}
\def\Unif{{\sf Unif}}

\def\de{{\rm d}}

\def\cG{\mathcal{G}}
\def\d{{\mathrm{d}}}

\def\Var{{\sf Var}}

\newcommand\norm[1]{\left\lVert{#1}\right\rVert}
\newcommand\abs[1]{\left\lvert{#1}\right\rvert}

\newcommand\myeqref[1]{{Eq.\,\eqref{#1}}}

\def\one{{\bf 1}}

\def\bb{{\bf b}}

\def\bG{{\bf G}}

\def\tG{\widetilde{G}}

\def\Var{{\rm Var}}

\def\Pair{\binom{[n]}{2}}
\def\bDelta{{\boldsymbol \Delta}}

\def\bs{{\boldsymbol s}}
\def\bx{{\boldsymbol x}}
\def\by{{\boldsymbol y}}
\def\bv{{\boldsymbol v}}
\def\bz{{\boldsymbol z}}
\def\bu{{\boldsymbol u}}

\def\bX{{\boldsymbol X}}

\def\bZ{{\boldsymbol Z}}
\def\bY{{\boldsymbol Y}}
\def\bG{{\boldsymbol G}}
\def\btG{\boldsymbol{\widetilde{G}}}
\def\bV{{\boldsymbol V}}

\def\bR{{\boldsymbol R}}

\def\beps{{\boldsymbol{\eps}}}
\def\ba{{\boldsymbol a}}
\def\bb{{\boldsymbol b}}

\DeclareMathAlphabet{\mathpzc}{OT1}{pzc}{m}{it}

\def\BEC{{\sf BEC}}

\def\Info{{\sf I}}
\def\hx{\widehat{x}}
\def\hbx{{\boldsymbol{\widehat{x}}}}
\def\hs{\widehat{x}}
\def\hbs{{\boldsymbol{\widehat{s}}}}
\def\bxi{{\boldsymbol{\xi}}}

\def\mmse{{\sf mmse}}
\def\MMSE{{\sf MMSE}}
\def\MSEAMP{{\sf MSE}_{\sf AMP}}
\def\vmse{{\sf vmmse}}
\def\overlap{{\sf Overlap}}

\def\tG{\widetilde{G}}

\def\ons{{\sf b}}
\def\tons{{\sf \tilde{b}}}
\def\bone{{\mathbf{1}}}

\def\sbmH{{\cal{H}_\mathrm{SBM}}}

\def\hcube{{\{\pm 1\}^n}}

\def\mseAMP{{\sf mse_{\sf AMP}}}

\def\Binom{{\rm Binom}}
\def\SBM{{\rm SBM}}

\def\op{\overline{p}}
\def\Unif{{\rm Uniform}}

\def\sech{{\text{sech}}}

\begin{document}

\title{Asymptotic Mutual Information for the\\ Two-Groups Stochastic Block Model}

\date{\today}
\author{Yash Deshpande\footnote{Department of Electrical
    Engineering, Stanford University},  \;\; Emmanuel Abbe\footnote{Department of Electrical
    Engineering and Program in Applied and Computational Mathematics, Princeton University}, 
\;\; Andrea Montanari\footnote{Department of Electrical
    Engineering and Department of Statistics, Stanford University}}

\maketitle

\begin{abstract}
We develop an information-theoretic view of the stochastic block
model, a popular statistical model for the large-scale structure of
complex networks. 
A graph $G$ from such a model is generated by first assigning vertex
labels at random from a finite alphabet, and then connecting vertices 
with edge probabilities depending on the labels of the endpoints.
In the case of the symmetric two-group model, we
establish an explicit `single-letter' characterization of the
per-vertex mutual information between the vertex labels and the graph.

The explicit expression of the mutual information is intimately 
related to estimation-theoretic quantities, and --in particular--
reveals a phase transition at the critical point for community detection. Below
the critical point the per-vertex mutual information is asymptotically the same
as if edges were independent. Correspondingly, no algorithm can
estimate the partition better than random guessing. 
Conversely, above the threshold, the per-vertex  mutual information is
strictly smaller than the independent-edges upper bound. 
In this regime there exists a procedure that estimates the vertex
labels better than random guessing.
\end{abstract}

\section{Introduction and main results}

The stochastic block model is the  simplest statistical model for 
networks with a community (or cluster) structure. As such, it has
attracted considerable amount of work across statistics, machine
learning, and theoretical computer science \cite{holland,dyer,snij,condon,mixed-core}. A random graph $\bG=(V,E)$ 
from this model has its vertex set $V$ partitioned into $r$
groups, which are assigned $r$ distinct labels. The probability of edge $(i,j)$ being present depends on the
group labels of vertices $i$ and $j$. 

In the context of social network analysis, groups correspond to social
communities \cite{holland}. For other data-mining applications, they
represent latent attributes of the nodes \cite{mcsherry}. In all of these
cases, we are interested in inferring the vertex labels from a single
realization of the graph. 

In this paper we develop an information-theoretic viewpoint on the
stochastic block model. Namely, we develop an explicit (`single-letter') 
expression for the per-vertex conditional entropy of the vertex labels given 
the graph. Equivalently, we compute the asymptotic per-vertex mutual information
between the graph and the vertex labels.
Our results hold asymptotically for large networks under
suitable conditions on the model parameters. 
The asymptotic mutual information is of independent interest, but is
also intimately related to estimation-theoretic quantities.

For the sake of simplicity, we will focus on the symmetric two group
model. Namely,
we assume the vertex set $V=[n]\equiv \{1,2,\dots,n\}$ to be
partitioned  into two sets $V = V_+\cup V_-$, with $\prob(i\in V_+) =
\prob(i\in V_-) = 1/2$ independently across vertices $i$. In
particular, the size of each group $|V_+|,|V_-|\sim \Binom(n,1/2)$
concentrates tightly around its expectation $n/2$.
Conditional on the edge labels, edges are independent with
\begin{align}
\prob\big((i,j)\in E\big|V_+,V_-\big)  = \begin{cases}
p_n & \mbox{ if $\{i,j\}\subseteq V_+$, or $\{i,j\}\subseteq V_-$,}\\
q_n& \mbox{ otherwise.}\\
\end{cases}
\end{align}
Throughout we will denote by $\bX = (X_i)_{i\in V}$ the set of vertex
labels $X_i\in \{+1,-1\}$, and we will be interested in the
conditional entropy $H(\bX|\bG)$ or --equivalently-- the mutual
information $I(\bX;\bG)$ in the limit $n\to\infty$.
We will write $\bG\sim \SBM(n;p,q)$  (or $(\bX,\bG)\sim\SBM(n;p,q)$) to
imply that the graph $G$ is distributed according to the stochastic
block model with $n$ vertices and parameters $p,q$.

Since we are interested in the large $n$ behavior, two preliminary
remarks are in order:
\begin{enumerate}
\item \emph{Normalization.} We obviously have\footnote{Unless
    explicitly stated otherwise, logarithms will be in base $e$, and
    entropies will be measured in nats.} $0\le H(\bX|\bG)\le n\,\log
  2$. It is therefore natural to study the per-vertex entropy
  $H(\bX|\bG)/n$.

 As we will see, depending on the model parameters, 
this will take any value between $0$ and $\log 2$. 
\item \emph{Scaling.} The reconstruction problem becomes easier when
  $p_n$ and $q_n$ are well separated, and more difficult when they are
  closer to each other.
For instance, in an early contribution, Dyer and Frieze \cite{dyer} proved
that the labels can be reconstructed exactly --modulo an overall
flip-- if $p_n=p>q_n=q$ are
distinct and independent of $n$. This --in particular-- implies 
$H(\bX|\bG)/n\to 0$ in this limit (in fact, it implies $H(\bX|\bG)\to \log 2$).
In this regime, the `signal' is so strong that the conditional entropy
is trivial. Indeed, recent work \cite{abh, mossel-consist} show that
this can 
also happen with $p_n$ and $q_n$ vanishing, and  characterizes the 
sequences $(p_n,q_n)$ for which this happens.
 (See Section
  \ref{sec:related} for an account of related work.)

Let $\op_n = (p_n+q_n)/2$ be the average edge probability.
It turns out that the relevant `signal-to-noise ratio' (SNR) is given by the
following parameter:
\begin{align}
\lambda_n = \frac{n\, (p_n-q_n)^2}{4\op_n(1-\op_n)}\, .
\end{align}
Indeed, we will see that $H(\bX|\bG)/n$ of order $1$, and has a strictly
positive limit when $\lambda_n$ is of order one. This is also the
regime in which the fraction of incorrectly labeled vertices has a
limit that is strictly between $0$ and $1$.   
\end{enumerate}

\subsection{Main result: Asymptotic per-vertex mutual information}

As mentioned above, our main result provides a single-letter
characterization for the per-vertex mutual information. 
This is given in terms of an 
\emph{effective Gaussian
scalar channel}.
Namely, define the Gaussian channel
\begin{align}
  Y_0 = Y_0(\gamma) = \sqrt{\gamma} \, X_0 + Z_0, \label{eq:scalarprob_1}
\end{align}
where $X_0\sim \Unif(\{+1,-1\})$ independent\footnote{
Throughout the paper, we will generally denote scalar equivalents of vector/matrix
quantities with the $0$ subscript} of $Z_0\sim\normal(0,1)$.
We denote by $\mmse(\gamma)$ and $\Info(\gamma)$ the corresponding
minimum mean square error and
mutual information:
\begin{align}
\Info(\gamma) & = \E\,\log \Big\{\frac{\de p_{Y|X}(Y_0(\gamma)|X_0)}{\de p_{Y}(Y_0(\gamma))}\Big\}\, ,\label{eq:InfoDef}\\
  \mmse(\gamma ) &= \E\left\{ ( X_0- \E\left\{ X_0| Y_0(\gamma)\right\})^2
                   \right\}\, . \label{eq:mmseDef}
\end{align}
In the present case, these quantities can be written explicitly as
Gaussian integrals of elementary functions:
\begin{align}
\Info(\gamma) & = \gamma-\E\log\cosh\big(\gamma+\sqrt{\gamma}\,
Z_0\big)\, ,\label{eq:InfoFormula}\\
\mmse(\gamma ) &= 1- \E\big\{\tanh(\gamma+\sqrt{\gamma}\, Z_0)^2\big\}\, . \label{eq:mmseFormula}
\end{align}

We are now in position to state our main result.
\begin{thm}\label{thm:main}
  For any $\lambda>0$, let $\gamma_* = \gamma_*(\lambda)$ be the
  largest non-negative solution of the equation:
  \begin{align}
      \gamma &= \lambda\big(1- \mmse(\gamma)\big)\, . \label{eq:MainEquation}
  \end{align}
We refer to $\gamma_*(\lambda)$ as to the \emph{effective
  signal-to-noise ratio}.
  Further, define $\Psi(\gamma, \lambda)$ by:
  \begin{align}
    \Psi(\gamma, \lambda ) &=
    \frac{\lambda}{4}+\frac{\gamma^2}{4\lambda}-\frac{\gamma}{2}+\Info(\gamma)\, .
\end{align}
Let the graph $\bG$ and vertex labels $\bX$ be distributed according to
the stochastic block model with $n$ vertices and parameters $p_n,q_n$
(i.e. $(\bG,\bX)\sim\SBM(n;p_n,q_n)$) and define $\lambda_n \equiv n\,
(p_n-q_n)^2/(4\op_n(1-\op_n))$.

Assume that, as $n\to\infty$, $(i)$ $\lambda_n\to\lambda$ and $(ii)$  $n\op_n(1-\op_n)\to\infty$.
Then, 
 \begin{align}
    \lim_{n\to\infty}  \frac{1}{n}\, I(\bX;\bG)&= \Psi(\gamma_*(\lambda),
    \lambda)\, .
\end{align}
\end{thm}

A few remarks are in order.

\begin{remark}
Of course, we could have stated our result in terms of conditional
entropy. Namely
 \begin{align}
    \lim_{n\to\infty}  \frac{1}{n}\, H(\bX|\bG)&= \log 2- \Psi(\gamma_*(\lambda),
    \lambda)\, .
\end{align}
\end{remark}

\begin{remark}
Notice that our assumptions require $n\op_n(1-\op_n)\to\infty$ \emph{at
  any, arbitrarily slow, rate}. In words, this corresponds to the graph
average degree diverging at any, arbitrarily slow, rate.

Recently (see Section \ref{sec:related} for a discussion of this literature), there has been considerable interest in the case of bounded
average degree, namely
\begin{align}
p_n  = \frac{a}{n}\, ,\;\;\;\;\;\; q_n  = \frac{b}{n}\, ,
\end{align}
with $a,b$ bounded.
Our proof gives an explicit error bound in terms of problem parameters
even when $n\op_n(1-\op_n)$ is of order one. Hence we are able to 
characterize 
the asymptotic mutual information for large-but-bounded average
degree up to an offset that vanishes with the average degree.

Explicitly, we prove that:
\begin{align}
  \limsup_{n\to\infty}
\Big|\frac{1}{n}I(\bX;\bG)-\Psi(\gamma_*(\lambda),\lambda) \Big| \le \frac{C\lambda^3}{\sqrt{a+b}}\, ,
\end{align}
for some absolute constant $C$.
\end{remark}

Our main result and its proof has implications on the minimum error
that can be achieved in estimating the labels $\bX$ from the graph $\bG$.
For reasons that will become clear below, a natural metric is given by the
matrix minimum mean square error
\begin{align}
\MMSE_n(\lambda) &\equiv
\frac{1}{n(n-1)}\E\Big\{\big\|\bX\bX^{\sT}-\E\{\bX\bX^{\sT}|\bG\}\big\|_F^2\Big\}\,. \label{eq:MMSEdef}
\end{align}
(Occasionally, we will also use the notation $\MMSE(\lambda;n)$ for
$\MMSE_n(\lambda)$.)
Using the exchangeability of the indices $\{1,\dots,n\}$, this can
also be rewritten as
\begin{align}
\MMSE_n(\lambda) &=\frac{2}{n(n-1)}\sum_{1\le i<j\le n}\E\big\{\big[X_iX_j-\E\{X_iX_j|\bG\}\big]^2\big\}\\
& = \E\big\{\big[X_1X_2-\E\{X_1X_2|\bG\}\big]^2\big\}\label{eq:MMSE_2vert}\\
& = \min_{\hx_{12}: \cG_n\to \reals}
\E\big\{\big[X_1X_2-\hx_{12}(\bG)\big]^2\big\}\, . \label{eq:MMSE_2vert_opt}
\end{align}
(Here $\cG_n$ denotes the set of graphs with vertex set $[n]$.)
In words, $\MMSE_n(\lambda)$ is the minimum error incurred in
estimating the relative sign of the labels of two given (distinct)
vertices. Equivalently, we can assume that vertex $1$ has label
$X_1=+1$. Then $\MMSE_n(\lambda)$ is the minimum mean square error
incurred in estimating the label of any other vertex, say vertex
$2$. Namely, by symmetry, we have (see Section \ref{sec:Estimation})
\begin{align}
\MMSE_n(\lambda)& = \E\big\{\big[X_2-\E\{X_2|X_1=+1,\bG\}\big]^2\big|X_1=+1\big\}\label{eq:AlternativeMMSE1}\\
& =  \min_{\hx_{2|1}: \cG_n\to\reals}
\E\big\{\big[X_2-\hx_{2|1}(\bG)\big]^2|X_1=+1\big\}\, . \label{eq:AlternativeMMSE2}
\end{align}
In particular $\MMSE_n(\lambda)\in [0,1]$, with  $\MMSE_n(\lambda)\to
1$ corresponding to random guessing.
\begin{thm}\label{thm:MainEstimation}
Under the assumptions of Theorem \ref{thm:main} (in particular
assuming $\lambda_n\to \lambda$ as $n\to\infty$), the following limit 
holds for the matrix minimum mean square error 
\begin{align}
\lim_{n\to\infty} \MMSE_n(\lambda_n) =
1-\frac{\gamma_*(\lambda)^2}{\lambda^2}\, .
\end{align}
Further, this implies $\lim_{n\to\infty} \MMSE_n(\lambda_n) = 1$ for
$\lambda\le 1$ and $\lim_{n\to\infty} \MMSE_n(\lambda_n) < 1$ for
$\lambda>1$. 
\end{thm}
For further discussion of this result and its generalizations, we
refer to Section \ref{sec:Estimation}.
In particular, Corollary \ref{coro:Metrics} establishes that
$\lambda=1$ is a phase transition for other estimation metrics as
well, in particular for overlap and vector mean square error.

\begin{remark}
As Theorem \ref{thm:main}, also the last theorem holds under the mild
condition that the average degree $n\op_n$ diverges \emph{at any,
  arbitrarily slow rate}. This should be contrasted with the phase
transition of naive spectral methods.

It is well understood that the community structure can be estimated by
the principal eigenvector of the centered adjacency matrix
$\bG-\E\{\bG\} = (\bG-\op_n\bone\bone^{\sT})$. (We denote by $\bG$ the
graph as well as its adjacency matrix.) This approach is
successful fro $\lambda>1$ but requires average degree $n\op_n\ge
(\log n)^c$ for $c$ a constant \cite{capitaine2009largest,benaych2011eigenvalues}. 
\end{remark}

\begin{remark}
Our proof of Theorem \ref{thm:main} and Theorem
\ref{thm:MainEstimation}
involves the analysis of a Gaussian observation model, whereby the
rank one matrix $\bX\bX^{\sT}$ is corrupted by additive Gaussian noise,
according to   $\bY= \sqrt{\lambda/n} \bX\bX^\sT + \bZ$.
In particular, we prove a single letter characterization of the
asymptotic  mutual information per dimension in this model
$\lim_{n\to\infty}I(\bX;\bY)$, cf. Theorem
\ref{thm:singleletterGauss} below. The resulting asymptotic value is proved
to coincide with the asymptotic value in the stochastic block model,
as established in Theorem \ref{thm:main}. In other words, the
per-dimension mutual information turns out to be \emph{universal} across
multiple noise models.
\end{remark}

\subsection{Outline of the paper}

In Section \ref{sec:related} we review the literature on this
problem. We then discuss the connection with estimation in Section
\ref{sec:Estimation}. This section also demonstrates how to evaluate the
asymptotic formula in Theorem \ref{thm:main}.

Section \ref{sec:Strategy} describes the proof strategy. As an
intermediate step, we introduce a Gaussian observation model which is
of independent interest. The proof of Theorem \ref{thm:main} is
reduced to two main propositions:
\begin{itemize}
\item Proposition  \ref{prop:gausseq} establishes that --within the
  regime defined in  Theorem \ref{thm:main}-- the stochastic block
  model is asymptotically equivalent to the Gaussian observation
  model (see Section \ref{sec:Strategy} for a formal definition).
  This statement (with explicit error bounds) is proved in
  Section \ref{sec:Gausseq} through a careful application of the
  Lindeberg method.
\item Proposition \ref{prop:singleletter} develops a single-letter
  characterization of the asymptotic per-vertex mutual information
of the Gaussian observation model. The proof of this fact is presented
in Section \ref{sec:singleletterproof} and builds on two steps. We
first prove an asymptotic upper bound on the matrix minimum mean
square error $\MMSE_n(\lambda)$ using an approximate message passing (AMP)
algorithm. We then use an area theorem to prove that this upper bound
is tight.
\end{itemize}

Finally, Section \ref{sec:ProofEstimation} contains the proof of
Theorem  \ref{thm:MainEstimation}.
Several technical details are deferred to the appendices.

\subsection{Notations}

The set of first $n$ integers is denoted by 
$[n] = \{1,2,\dots,n\}$. 

When possible, we will follow the convention of denoting random
variables by upper-case letters (e.g. $X,Y,Z,\dots$), and their values by lower case
letters (e.g. $x,y,z,\dots$). 
We use boldface for vectors and matrices, e.g. $\bX$ for a random
vector and $\bx$ for a deterministic vector. The graph $\bG$ will be
identified with its adjacency matrix.  Namely, with a slight abuse of
notation, we will use $\bG$ both to denote a graph $\bG=(V=[n], E)$
(with $V=[n]$ the vertex set, and $E$ the edge set, i.e. a set of
unordered pairs of vertices), and its adjacency matrix. 
This is a symmetric zero-one matrix $\bG =
(G_{ij})_{1\le i,j\le n}$ with entries
\begin{align}
G_{ij} = \begin{cases}
1 &\; \mbox{if $(i,j)\in E$},\\
0 &\; \mbox{otherwise.}
\end{cases}
\end{align}
Throughout we assume $G_{ii} =0$ by convention.

We write $f_1(n) = f_2(n) + O(f_3(n))$ to mean that $\abs{f_1(n) - f_2(n)}\le Cf_3(n) $ for
a universal constant $C$. We denote by $C$ a generic (large)
constant that is independent of problem parameters, whose value
can change from line to line.

We say that an event holds \emph{with high probability} if it holds
with probability converging to one as $n\to\infty$.

We denote the $\ell_2$ norm of a vector $\bx$ by $\norm{\bx}_2$ and the 
Frobenius norm of a matrix $\bY$ by $\norm{\bY}_F$. The ordinary
scalar product of vectors $\ba,\bb\in\reals^m$ is denoted as
$\<\ba,\bb\> \equiv \sum_{i=1}^ma_ib_i$.

Unless stated otherwise, logarithms will be taken in the natural
basis, and entropies measured in nats.

%
%*****************************************************************************************
%

\section{Related work} \label{sec:related}
The stochastic block model was first introduced within the social
science literature in \cite{holland}. Around the same time, it was
studied within theoretical computer science \cite{bui,dyer}, under the
name of `planted partition model. '

A large part of the literature has focused on the problem of 
\emph{exact recovery} of the community (cluster) structure.
A long series of papers
\cite{bui,dyer,boppana,snij,jerrum,condon,carson,mcsherry,bickel,rohe,choi,sbm-algos,Vu-arxiv,chen-xu}, 
establishes sufficient conditions on the gap between $p_n$ and $q_n$ that
guarantee exact recovery of 
the vertex labels with  high probability.
A sharp threshold for exact recovery was  obtained in
\cite{abh,mossel-consist}, showing that for 
$p_n=\alpha \log(n)/n$, $q_n= \beta\log(n)/n$, $\alpha,\beta>0$, exact recovery is solvable
(and efficiently so) if and only if $\sqrt{\alpha}- \sqrt{\beta} \geq 2$. 
Efficient algorithms for this problem were also developed in \cite{prout,new-xu, afonso_lap}.
For the SBM with arbitrarily many communities, necessary and sufficient
conditions for exact recovery were recently obtained in 
\cite{colin1}. The resulting sharp threshold is efficiently
achievable and is stated in terms  of a CH-divergence. 

A parallel line of work studied the \emph{detection} problem.
In this case, the estimated community structure is only required to be
asymptotically positively correlated with the ground truth.
For this requirement, two independent groups
\cite{massoulie-STOC,Mossel_SBM2} proved  that detection is solvable
(and so efficiently) if and only if $(a-b)^2 > 2(a+b)$, when
$p_n=a/n$, $q_n=b/n$.
This settles  a conjecture made in \cite{decelle} and improves  on
earlier work \cite{coja-sbm}. Results for detection with more than two
communities were recently obtained in 
\cite{sbm-groth,new-vu,colin1,bordenave2015non}. 
A variant of community detection with a single hidden community in a sparse graph was studied in \cite{am_1comm}. 

In a sense, the present paper bridges detection and exact recovery, by
characterizing the minimum estimation error when this is non-zero,
but --for $\lambda>1$-- smaller than for random guessing.

An information-theoretic view of the SBM was first introduced in
\cite{random,abbetoc}. 
There it was shown that in the regime of $p_n=a/n$, $q_n=b/n$, and $a
\leq b$ (i.e., disassortative communities), 
the normalized mutual information $I(\bX;\bG)/n$ admits a limit as
$n\to\infty$. This result is obtained by showing that the condition
entropy 
$H(\bX|\bG)$ is sub-additive in $n$, using an interpolation method for
planted models. While the result of \cite{random,abbetoc} 
holds for arbitrary $a\leq b$ (possibly small) and extend to a broad
family of planted models, 
the existence of the limit in the assortative case $a>b$ is left open.
Further, sub-additivity methods do not provide any insight as to the
limit value.

For the partial recovery of the communities, it was shown in
\cite{mossel-consist} that the communities 
can be recovered up to a vanishing fraction of the nodes if and only
if $n(p-q)^2/(p+q)$ diverges. 
This is generalized in \cite{colin1} to the case of more than two communities.
In these regimes, the normalized mutual information $I(\bX;\bG)/n$ (as
studied in this paper) tends to $\log 2$ nats. 
For the constant degree regime, it was shown in \cite{mossel2} that
when 
$(a-b)^2/(a+b)$ is sufficiently large, the fraction of nodes that can
be recovered is determined by the broadcasting problem on tree
\cite{evans}. 
Namely, consider the reconstruction problem whereby a bit is
broadcast on a Galton-Watson tree with  Poisson($(a+b)/2$) offspring
and with binary symmetric channels of bias $b/(a+b)$ on each branch.
Then the probability of recovering the bit correctly from the leaves
at large depth gives the fraction of nodes that can be correctly
labeled in the SBM.

In terms of proof techniques, our arguments are closest to \cite{korada2011applications,deshpande2014information}. We use the well-known Lindeberg strategy to reduce computation of
mutual information in the SBM to mutual information of the Gaussian
observation model. We then compute  the latter mutual information by 
developing sharp algorithmic upper bounds, which are then shown to be asymptotically
tight via an area theorem. 
The Lindeberg strategy builds from \cite{korada2011applications,
chatterjee2006generalization} while the area theorem argument also 
appeared in \cite{montanari2006analysis}. We expect these techniques to
be more broadly applicable to compute quantities like normalized mutual
information or conditional entropy in a variety of models.

Let us finally mentioned that the result obtained in this paper are
likely to extend to more general SBMs, 
with multiple communities, to the Censored Block Model studied in
\cite{abbetoc,abbs,Huang_Guibas_Graphics,Chen_Huang_Guibas_Graphics,Chen_Goldsmith_ISIT2014,abbs-isit,rough,new-xu,new-vu,florent_CBM},
the Labeled Block Model \cite{label_marc,jiaming}, 
and other variants of block models. In particular, it would be
interesting to understand which estimation-theoretic quantities 
appear for these models, and whether a general result stands 
behind the case of this paper. 

While this paper was in preparation, Lesieur, Krzakala and Zdborov\'a
\cite{lesieur2015mmse} studied estimation of low-rank matrices
observed through noisy memoryless channels. They conjectured that the
resulting minimal estimation error is universal across a variety of
channel models. Our proof  (see Section \ref{sec:Strategy} below)
establishes universality across two such models: the Gaussian and the
binary output channels. We expect that similar techniques can be
useful to prove universality for other models as well.

%
%*******************************************
%

\section{Estimation phase transition}
\label{sec:Estimation}

In this section we discuss how to evaluate the asymptotic formulae in 
Theorem \ref{thm:main} and Theorem \ref{thm:MainEstimation}.
We then discuss the consequences of our results for various
estimation metrics.

Before passing to these topics, we will derive a simple upper bound on
the per-vertex mutual information,
which will be a useful comparison for our results.

\subsection{An elementary upper bound}

It is instructive to start with an elementary upper bound on
$I(\bX;\bG)$.
\begin{lemma}\label{lemma:Elementary}
Assume $p_n$, $q_n$ satisfy the assumptions of Theorem \ref{thm:main}
(in particular $(i)$ $\lambda_n\to\lambda$ and $(ii)$
$n\op_n(1-\op_n)\to \infty$). Then
\begin{align}
\limsup_{n\to\infty}\frac{1}{n}I(\bX;\bG)\le \frac{\lambda}{4}\, .
\end{align}
\end{lemma}
\begin{proof}
We have
\begin{align}
\frac{1}{n}\, I(\bX;\bG) &= \frac{1}{n}\, H(\bG)-\frac{1}{n}\,
  H(\bG|\bX) \\
&\stackrel{(a)}{=}  \frac{1}{n}\, H(\bG)-\frac{1}{n}\sum_{1\le i<j\le n}
  H(G_{ij}|\bX)\\
&\stackrel{(b)}{=}  \frac{1}{n}\, H(\bG)-\frac{1}{n}\sum_{1\le i<j\le n}
  H(G_{ij}|X_i\cdot X_j)\\
&\le \frac{1}{n}\sum_{1\le i<j\le n} I(X_i\cdot X_j;G_{ij}) =
  \frac{n-1}{2}\, I(X_1\cdot X_2;G_{12}) \, ,
\end{align} 
where $(a)$ follows since $\{G_{ij}\}_{i<j}$ are conditionally
independent given $\bX$ and $(b)$ because $G_{ij}$ only depends on
$\bX$ through the product $X_i\cdot X_j$ (notice that there is no
comma but product in $H(G_{ij}|X_i\cdot X_j)$.

From our model, it is easy to check that
\begin{align}
I(X_1\cdot X_2; G_{12}) = \frac{1}{2}\, p_n\log\frac{p_n}{\op_n} +
  \frac{1}{2}\, q_n\log\frac{q_n}{\op_n}+\frac{1}{2}\, (1-p_n)\log\frac{1-p_n}{1-\op_n} +
  \frac{1}{2}\, (1-q_n)\log\frac{1-q_n}{1-\op_n}\, .
\end{align}
The claim follows by substituting $p_n =
\op_n+\sqrt{\op_n(1-\op_n)\lambda_n/n}$,
$q_n = \op_n-\sqrt{\op_n(1-\op_n)\lambda_n/n}$ and by Taylor
expansion\footnote{Indeed Taylor expansion yields the stronger result 
$n^{-1}I(\bX;\bG) \le  (\lambda_n/4)+n^{-1}$ for all $n$ large enough.}.
\end{proof}

\subsection{Evaluation of the asymptotic formula}

\begin{figure}[t!]
\hspace{3.5cm}\includegraphics[scale=0.5]{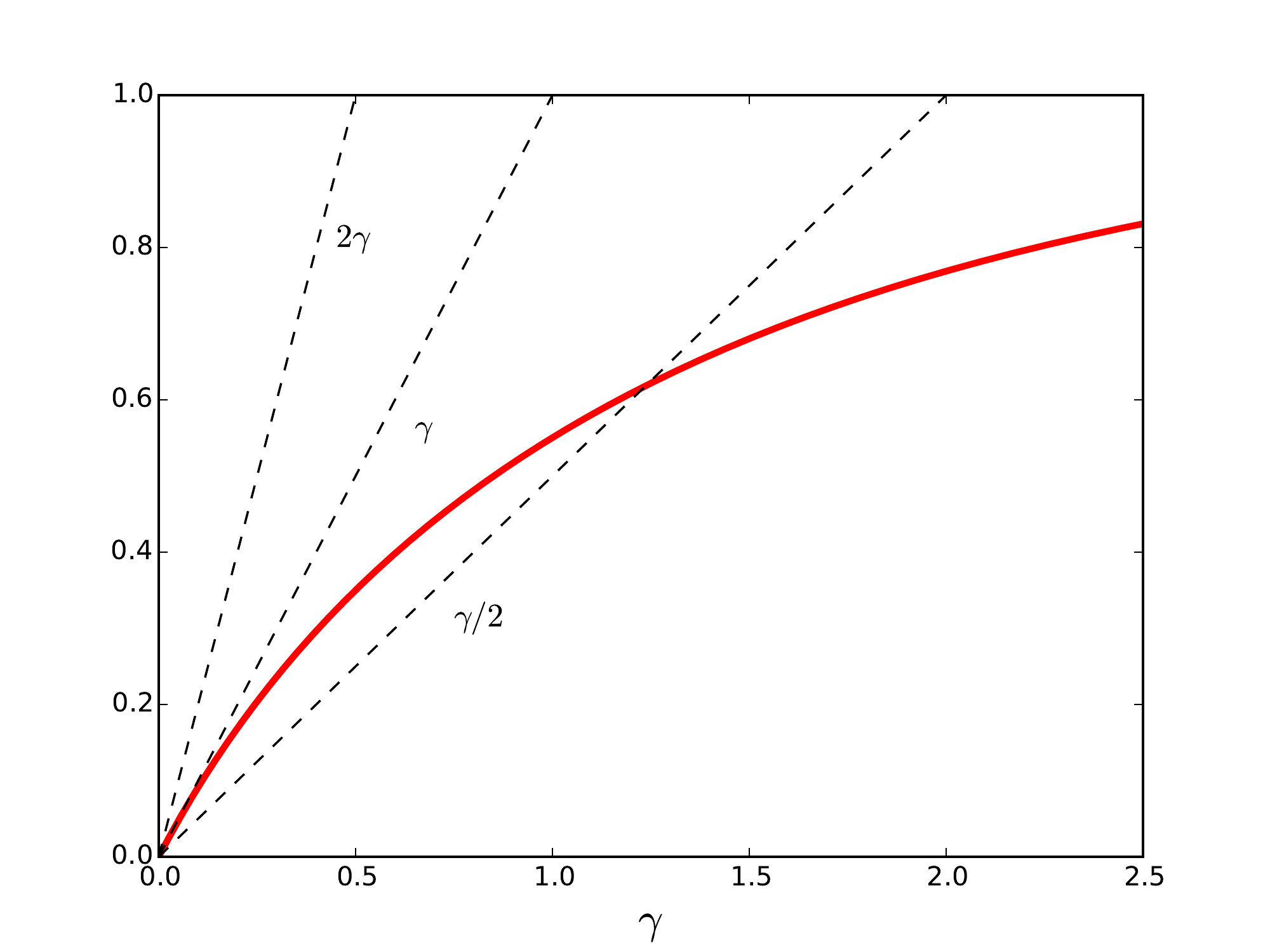}
\put(-90,120){$1-\mmse(\gamma)$}
 \caption{Illustration of the fixed point equation
   Eq.~(\ref{eq:MainEquation}). The `effective signal-to-noise ratio'
   $\gamma_*(\lambda)$ is given by the intersection of the curve
$\gamma\mapsto G(\gamma) = 1-\mmse(\gamma)$, and the line $\gamma/\lambda$.}
  \label{fig:FormulaEvaluation}
\end{figure}

\begin{figure}[t!]
\includegraphics[scale=0.4]{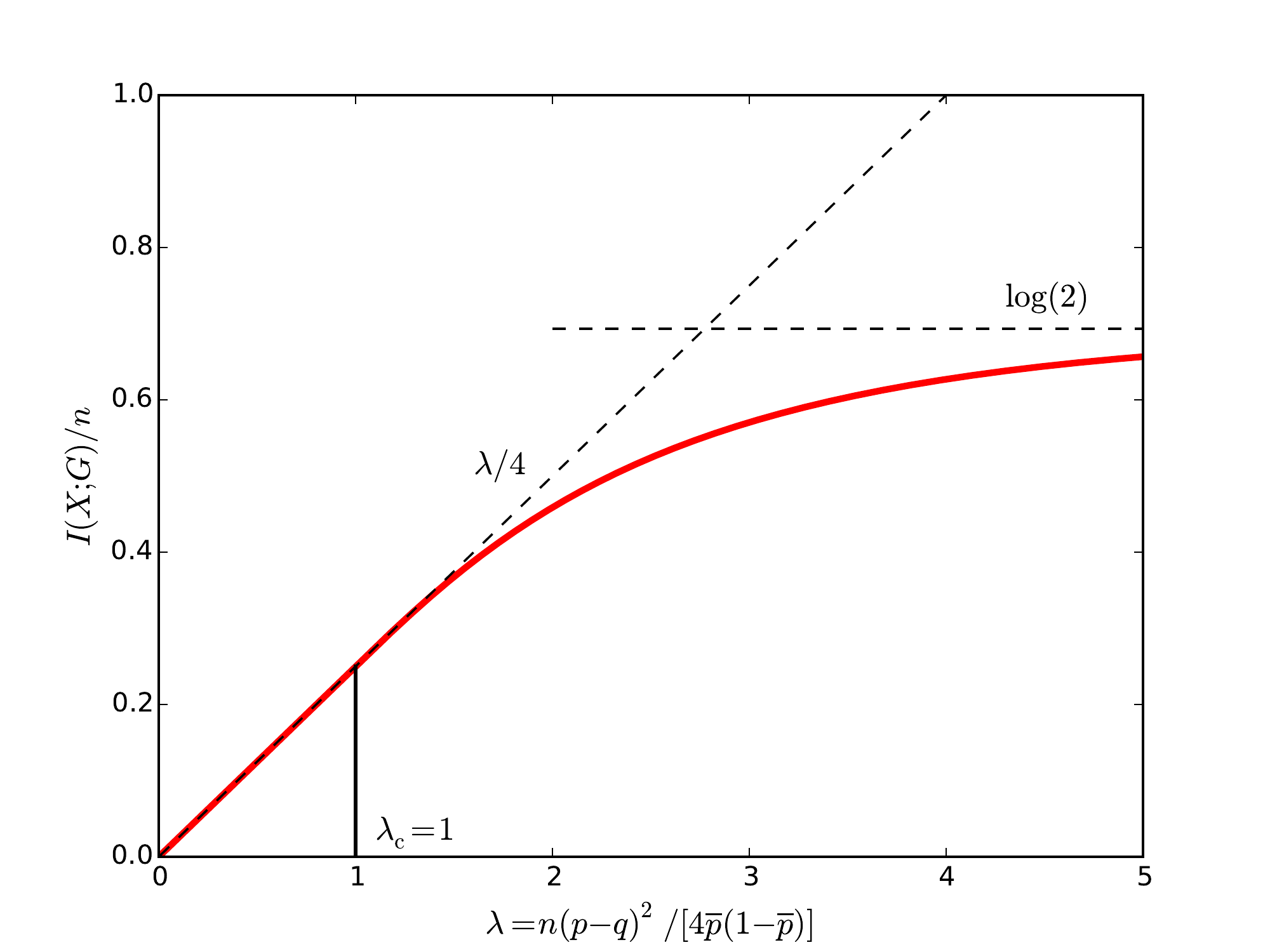}
 \includegraphics[scale=0.4]{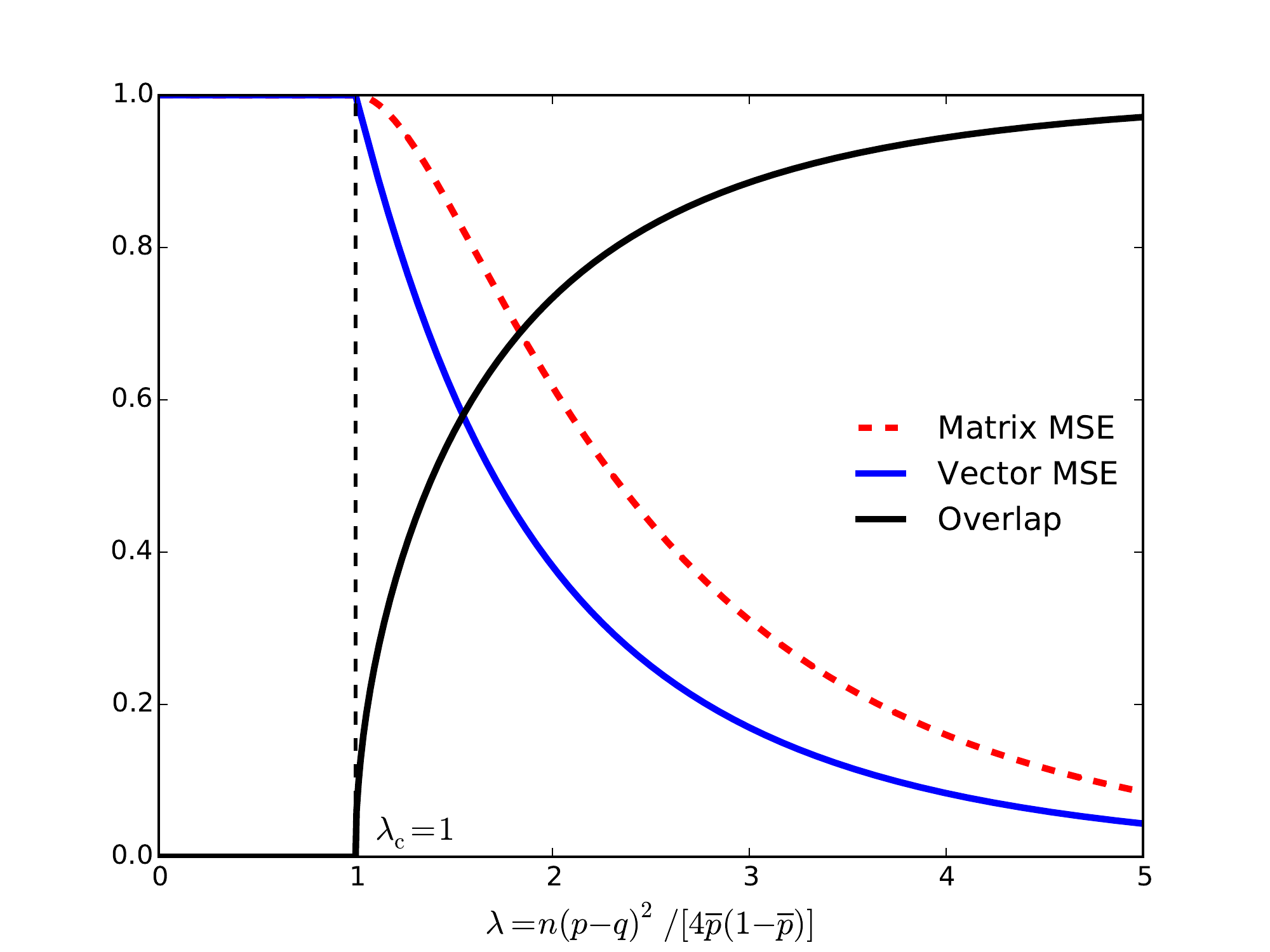}
 \caption{Left frame: Asymptotic mutual information per vertex of the
   two-groups stochastic block model, as a function of the
   signal-to-noise ratio $\lambda$. The dashed lines are simple upper
   bounds: $\lim_{n\to\infty} I(\bX;\bG)/n\le \lambda/4$ (cf. Lemma
   \ref{lemma:Elementary}) and $I(\bX;\bG)/n\le \log 2$. Right frame:
Asymptotic estimation error under different metrics (see Section
\ref{sec:Est}). Note the phase transition at $\lambda=1$ in both
frames.}
  \label{fig:InformationEstimation}
\end{figure}

Our  asymptotic expression for the mutual information, cf. Theorem
\ref{thm:main}, and for the estimation error, cf. Theorem
\ref{thm:MainEstimation}, depends on the solution of
Eq.~(\ref{eq:MainEquation}) which we copy here for the reader's
convenience:
  \begin{align}
      \gamma &= \lambda\big(1- \mmse(\gamma)\big)  \; \equiv \lambda
      \, G(\gamma)\, .\label{eq:MainEquation2}
  \end{align}
Here we defined
\begin{align}
G(\gamma) = 1- \mmse(\gamma) = \E\big\{\tanh(\gamma+\sqrt{\gamma}\,
Z)^2\big\}\, . \label{eq:Gdef}
\end{align}
The  effective signal-to-noise ratio $\gamma_*(\lambda)$ that enters Theorem
\ref{thm:main} and Theorem \ref{thm:MainEstimation} is the largest non-negative solution of
Eq.~(\ref{eq:MainEquation}). 
This equation is illustrated in Figure
\ref{fig:FormulaEvaluation}.

It is immediate to show from the definition (\ref{eq:Gdef}) that $G(\,\cdot\,)$ is continuous on
$[0,\infty)$ with $G(0) = 0$, and $\lim_{\gamma\to\infty}G(\gamma) =
1$. This in particular implies that $\gamma = 0$ is always a solution
of Eq.~(\ref{eq:MainEquation}).
Further, since $\mmse(\gamma)$ is monotone decreasing in the
signal-to-noise ratio $\gamma$, $G(\gamma)$ is monotone increasing. 
As shown in the proof of Remark \ref{rem:unique} (see Appendix
\ref{sec:ProofUnique}), $G(\,\cdot\,)$ is also strictly concave on
$[0,\infty)$. This implies that Eq.~(\ref{eq:MainEquation}) as at most
one solution in $(0,\infty)$, and a strictly positive solution only
exists if $\lambda G'(0) = \lambda >1$.

We summarize these remarks below, and refer to Figure
\ref{fig:InformationEstimation} for an illustration.
\begin{lemma}\label{lemma:FixedPoint}
The effective SNR, and the asymptotic expression for the per-vertex
mutual information in Theorem \ref{thm:main} have the following
properties:
\begin{itemize}
\item For $\lambda\le 1$, we have $\gamma_*(\lambda) = 0$ and $\Psi(\gamma_*(\lambda),
    \lambda) = \lambda/4$.
\item For $\lambda\le 1$, we have $\gamma_*(\lambda) \in (0,\lambda)$
  strictly with $\gamma_*(\lambda)/\lambda\to 1$ as
  $\lambda\to\infty$. 

Further, $\Psi(\gamma_*(\lambda),
    \lambda) < \lambda/4$ strictly with $\Psi(\gamma_*(\lambda),
    \lambda)\to\log 2$ as $\gamma\to\infty$.
\end{itemize}
\end{lemma}
\begin{proof}
All of the claims follow immediately form the previous remarks, and
simple calculus, except the claim $\Psi(\gamma_*(\lambda),\lambda) <
\lambda/4$ for
$\lambda>1$. This is direct consequence of the variational
characterization established below.
\end{proof}

We next give an alternative (variational) characterization of the asymptotic formula
which is useful for proving bounds.
\begin{lemma}
Under the assumptions and definitions of Theorem \ref{thm:main}, we
have
\begin{align}
\lim_{n\to\infty}  \frac{1}{n}\, I(\bX;\bG)= \Psi(\gamma_*(\lambda),
    \lambda) =\min_{\gamma\in [0,\infty)} \, \Psi(\gamma,\lambda)\, .
\end{align}
\end{lemma}
\begin{proof}
The function $\gamma\mapsto\Psi(\gamma,\lambda)$ is differentiable on
$[0,\infty)$
with $\Psi(\gamma,\lambda) = \gamma^2/(4\lambda) +O(\gamma)\to\infty$
as $\gamma\to\infty$. Hence, the $\min_{\gamma\in [0,\infty)} \,
\Psi(\gamma,\lambda)$
is achieved at a point where the first derivative vanishes (or,
eventually, at $0$). 
Using the I-MMSE relation \cite{guo2005mutual}, we get
\begin{align}
\frac{\partial\Psi}{\partial\gamma}(\gamma,\lambda) =
\frac{\gamma}{2\lambda}-\frac{1}{2}+ \frac{1}{2}\, \mmse(\gamma)\, .\label{eq:DerPsi}
\end{align}
Hence the minimizer is a solution of Eq.~(\ref{eq:MainEquation}). 
As shown above, for $\lambda\le 1$, the only solution is
$\gamma_*(\lambda) = 0$, which therefore yields $\Psi(\gamma_*(\lambda),
    \lambda) =\min_{\gamma\in [0,\infty)} \, \Psi(\gamma,\lambda)$ as
    claimed.

For $\lambda>1$, Eq.~(\ref{eq:MainEquation}) admits the two solutions:
$0$ and $\gamma_*(\lambda)>0$. However, by expanding
Eq.~(\ref{eq:DerPsi}) for small $\gamma$, we obtain
$\Psi(\gamma,\lambda) = \Psi(0,\lambda) - (1-\lambda^{-1})\gamma^2/4 +
o(\gamma)$ and hence $\gamma=0$ is a local maximum, which implies the
claim for $\lambda>1$ as well.
\end{proof}

We conclude by noting that Eq.~(\ref{eq:MainEquation}) can be solved
numerically rather efficiently. The simplest method consists is by
iteration. Namely, we initialize $\gamma^0=\lambda$ and then iterate
$\gamma^{t+1} = \lambda\,G(\gamma^t)$. This approach was used for
Figure \ref{fig:InformationEstimation}.

\subsection{Consequences for estimation}
\label{sec:Est}

Theorem \ref{thm:MainEstimation} establishes that a phase transition
takes place at $\lambda=1$ for the matrix minimum mean square error $\MMSE_n(\lambda_n)$
defined in Eq.~(\ref{eq:MMSEdef}). 
Throughout this section, we will omit the subscript $n$ to denote the
$n\to\infty$ limit (for instance, we write $\MMSE(\lambda) \equiv
\lim_{n\to\infty}\MMSE_n(\lambda_n)$). 

Figure \ref{fig:InformationEstimation} reports the asymptotic
prediction  for $\MMSE(\lambda)$ stated in Theorem
\ref{thm:MainEstimation},
and evaluated as discussed above. The error decreases rapidly to $0$
for $\lambda>1$. 

In this section we discuss two other estimation metrics. In both cases
we define these metrics by optimizing a suitable risk over a class of 
estimators: it is understood that randomized estimators are admitted
as well.
\begin{itemize}
\item The first metric is the \emph{vector minimum mean square error}:
\begin{align}
\vmse_n(\lambda_n)  &= \frac{1}{n}\, \inf_{\hbx: \cG_n\to \reals^n}
                      \E\Big\{\min_{s\in\{+1,-1\}}\big\|\bX-s\,
                      \hbx(\bG)\big\|_2^2\Big\}\, . \label{eq:vmseDef}
\end{align}
Note the minimization over the sign $s$: this is necessary because the
vertex labels can be estimated only up to an overall flip.
Of course  $\vmse_n(\lambda_n)\in [0,1]$, since it is always possible
to achieve vector mean square error equal to one by returning
$\hbx(\bG) = 0$.   
\item The second metric is the \emph{overlap}:
\begin{align}
\overlap_n(\lambda_n)= \frac{1}{n}\sup_{\hbs: \cG_n\to
  \{+1,-1\}^n}\E\big\{|\<\bX,\hbs(\bG)\>|\big\}\, .\label{eq:OverlapDef}
\end{align}
Again $\overlap_n(\lambda_n)\in [0,1]$ (but now large overlap
corresponds to good estimation). Indeed by returning $\hs_i(\bG)
\in\{+1,-1\}$ uniformly at random, we obtain
$\E\big\{|\<\bX,\hbs(\bG)\>|\big\}/n = O(n^{-1/2})\to 0$.

Note that the main difference between overlap and vector minimum mean
square error is that in the latter case we consider estimators
$\hbx:\cG_n\to\reals^n$ taking arbitrary real values, while in the
former we assume estimators $\hbs:\cG_n\to\{+1,-1\}^n$ taking binary values.
\end{itemize}

In order to clarify the relation between various metrics, we begin by
proving the alternative characterization of the matrix minimum mean
square error  in Eqs.~(\ref{eq:AlternativeMMSE1}),
(\ref{eq:AlternativeMMSE2}).
\begin{lemma}
Letting $\MMSE_n(\lambda)$ be defined as per Eq.~(\ref{eq:MMSEdef}),
we have
\begin{align}
\MMSE_n(\lambda)& =
\E\big\{\big[X_2-\E\{X_2|X_1=+1,\bG\}\big]^2\big|X_1=+1\big\} \label{eq:AlternativeMMSE1_bis}\\
& =  \min_{\hx_{2|1}: \cG_n\to\reals}
\E\big\{\big[X_2-\hx_{2|1}(\bG)\big]^2|X_1=+1\big\}\, . \label{eq:AlternativeMMSE2_bis}
\end{align}
\end{lemma}
\begin{proof}
First note that Eq.~(\ref{eq:AlternativeMMSE2_bis}) follows
immediately from Eq.~(\ref{eq:AlternativeMMSE1_bis}) since conditional
expectation minimizes the mean square error (the conditioning only
changes the prior on $\bX$).

In order to prove Eq.~(\ref{eq:AlternativeMMSE1_bis}),
we start from Eq.~(\ref{eq:MMSE_2vert}). Since the prior distribution
on $X_1$ is uniform, we have
\begin{align}
\E\{X_1X_2|\bG\} &= \frac{1}{2} \E\{X_1X_2|X_1 = +1, \bG \}
+\frac{1}{2} \E\{X_1X_2|X_1 = +1, \bG \}  \\
& = \E\{X_2|X_1 = +1, \bG \}\, .
\end{align}
where in the second line we used the fact that, conditional to $\bG$,
$\bX$ is distributed as $-\bX$. 
Continuing from Eq.~(\ref{eq:MMSE_2vert}), we get
\begin{align}
\MMSE_n(\lambda) &=
\E\big\{\big[X_1X_2-\E\{X_2|X_1=+1,\bG\}\big]^2\big\}\\
&=
\frac{1}{2}\, \E\big\{\big[X_1X_2-\E\{X_2|X_1=+1,\bG\}\big]^2\big|X_1
= +1\big\}\nonumber\\
&\phantom{aaaa}+ 
\frac{1}{2}\, \E\big\{\big[X_1X_2-\E\{X_2|X_1=+1,\bG\}\big]^2\big|X_1
= -1\big\}\\
& =  \E\big\{\big[X_2-\E\{X_2|X_1=+1,\bG\}\big]^2\big|X_1
= +1\big\}\, ,
\end{align}
which proves the claim.
\end{proof}

The next lemma clarifies the relationship between matrix and vector
minimum mean square error. Its proof is deferred to Appendix \ref{sec:VectorVsMatrix}.
\begin{lemma}\label{lemma:VectorVsMatrix}
With the above definitions, we have
\begin{align}
1-\sqrt{1-(1-n^{-1})\MMSE_n(\lambda)} \le \vmse_n(\lambda) \le\MMSE_n(\lambda)\, .\label{eq:VmseBounds}
\end{align}
\end{lemma}
Finally, a lemma that relates overlap and vector minimum mean square error,
whose proof can be found in  Appendix \ref{sec:OverlapVsVector}.
\begin{lemma}\label{lemma:OverlapVsVector}
With the above definitions, we have
\begin{align}
\overlap_n(\lambda)\ge 1-\vmse_n(\lambda) - O(n^{-1/2})\, .\label{eq:OverlapBounds}
\end{align}
\end{lemma}

As an immediate corollary of these lemmas (together with Theorem
\ref{thm:MainEstimation}
and Lemma \ref{lemma:FixedPoint}), we obtain that $\lambda=1$ is the
critical point for other estimation metrics as well.
\begin{corollary}\label{coro:Metrics}
The vector minimum mean square error  and the overlap exhibit a \emph{phase transition
  at $\lambda=1$}. Namely, under the assumptions of Theorem \ref{thm:main}
(in particular, $\lambda_n\to\lambda$ and
$n\op_n(1-\op_n)\to\infty$), we have
\begin{itemize}
\item If $\lambda\le 1$, then estimation cannot be performed
  asymptotically better than without any information:
\begin{align}
\lim_{n\to\infty}\vmse_n(\lambda_n) &= 1\, ,\\
\lim_{n\to\infty}\overlap_n(\lambda_n) &= 0\, .
\end{align}
\item If $\lambda>1$, then estimation can be performed better 
than without any information, even in the limit $n\to\infty$:
\begin{align}
&0<1-\frac{\gamma_*(\lambda)}{\lambda}\le 
\liminf_{n\to\infty}\vmse_n(\lambda_n) \le
\limsup_{n\to\infty}\vmse_n(\lambda_n) \le
 1-\frac{\gamma_*(\lambda)^2}{\lambda^2} <1\, ,\\
&0<\frac{\gamma_*(\lambda)^2}{\lambda^2}\le 
\liminf_{n\to\infty}\overlap_n(\lambda_n) \, .
\end{align}
\end{itemize}
\end{corollary}
%
%*******************************************
%
\section{Proof strategy: Theorem \ref{thm:main}}
\label{sec:Strategy}

In this section we describe the main elements used in the proof
of Theorem \ref{thm:main}:
\begin{itemize}
\item We describe a Gaussian observation model which has
  asymptotically the same mutual information as the SBM introduced
  above.
\item We state an asymptotic characterization of the mutual
  information of this Gaussian model.
\item We describe an approximate message passing (AMP) estimation
  algorithm that plays a key role in the last characterization.
\end{itemize}
We then use these technical results (proved in later sections) to
prove Theorem \ref{thm:main} in Section \ref{sec:ProofTheoremMain}.

We recall that $\op_n = (p_n + q_n)/2$. Define the gap
$\Delta_n \equiv (p_n - q_n)/2 = \sqrt{\lambda_n \op_n
  (1-\op_n)/n}$. 
We will assume for the proofs that $\lim_{n\to \infty}\lambda_n = \lambda>0$ (i.e. the assortative
model) but the results
also hold for $\lambda<0$ in an analogous fashion.

%
%****************
%
\subsection{Gaussian model}

The edges $\{G_{ij}\}_{i<j}$ are conditionally independent given
the vertex labels $\bX$, with distribution:
\begin{align}
G_{ij} &= \begin{cases}
    1 &\text{ with probability } \op_n + \Delta_n X_iX_j \, ,\\
    0 &\text{ with probability } 1-\op_n - \Delta_n X_i X_j\, .
  \end{cases}
\end{align}
As a first step, we compare the SBM with an alternate Gaussian observation
model defined as follows. 
Let $\bZ$ be a Gaussian random symmetric matrix
generated with independent entries $Z_{ij}\sim\normal(0, 1)$ and
$Z_{ii}\sim \normal(0, 2)$, independent of $\bX$. Consider the noisy observations $\bY =
\bY(\lambda)$ defined by
\begin{align}
  \bY(\lambda) &= \sqrt\frac{\lambda}{n} \bX\bX^\sT + \bZ\, .
  \label{eq:gaussobs}
\end{align}
Note that this model matches the first two moments of the original
model. More precisely, if we define the rescaled adjacency matrix
$G^{\rm res}_{ij}\equiv (G_{ij}-\op_n)/\sqrt{\op_n(1-\op_n)}$, then
$\E\{G^{\rm res}_{ij}|\bX\} = \E\{Y_{ij}|\bX\}$ and $\Var(G^{\rm res}_{ij}|\bX) = \Var(Y_{ij}|\bX) +O(n^{-1/2})$.

 Our first proposition proves that the mutual 
 information between  the vertex labels $\bX$ and the observations 
agrees  to leading order across the two models.
\begin{propo}\label{prop:gausseq}
Assume that, as $n\to\infty$, $(i)$ $\lambda_n\to\lambda$ and $(ii)$
$n\op_n(1-\op_n)\to\infty$.
Then there is a constant $C$ independent of $n$ such that
  \begin{align}
    \frac{1}{n} \big|I(\bX;\bG) - I(\bX; \bY)\big| \le C\left(
    \frac{\lambda^{3/2}}{\sqrt{n\op_n(1-\op_n)}} + \abs{\lambda_n-\lambda}
    \right).
  \end{align}
\end{propo}
The proof of this result is presented in Section \ref{sec:Gausseq}.

The next step consists in analyzing the Gaussian model
(\ref{eq:gaussobs}), which is of independent interest.
It turns out to be convenient to embed this in a more general model
whereby, in addition to the observations $\bY$, we are also given
observations of $\bX$ through a binary erasure channel with erasure
probability $\beps=1-\eps$, $\BEC(\beps)$. We will denote by $\bX(\eps) = (X_1(\eps),\dots,X_n(\eps))$
the output of this channel, where we set $X_i(\eps) = 0$ every time
the symbol is erased. Formally we have
\begin{align}
X_{i}(\eps) = B_i\, X_i\, ,
\end{align}
where $B_i\sim \Ber(\eps)$ are independent random variables,
independent of $\bX$, $\bG$. In the special case $\eps=0$, all of
these observations are trivial, and we recover the original model. 

The reason for introducing the additional observations $\bX(\eps)$ is the following.
The graph $\bG$ has the same distribution conditional on $\bX$ or $-\bX$, hence
it is impossible to recover the sign of $\bX$. As we will see, the extra observations 
$\bX(\eps)$ allow to break this trivial symmetry and we will recover the required 
results by continuity in $\eps$ as the extra information vanishes.

Indeed, our next result establishes a
 single letter characterization of $I(\bX; \bY, \bX(\eps))$
in terms of a recalibrated \emph{scalar} observation problem.
Namely, we define the following observation model for $X_0\sim
\Unif(\{+1,-1\})$ a Rademacher random variable:
\begin{align}
  Y_0&= \sqrt{\gamma} X_0 + Z_0, \label{eq:scalarprob}\\
X_0(\eps) & = B_0X_0\, .
\end{align}
Here $X_0$, $B_0\sim \Ber(\eps)$,  $Z_0\sim\normal(0, 1)$, 
are mutually independent.  We denote by $\mmse( \gamma,\eps)$, the minimum
mean squared error of estimating $X_0$ from $X_0(\eps)$, $Y_0$, conditional
on $B_0$. 
Recall the definitions (\ref{eq:InfoDef}), (\ref{eq:mmseDef}) of
$\Info(\gamma)$, $\mmse(\gamma)$, and the
expressions (\ref{eq:InfoFormula}), (\ref{eq:mmseFormula}).
A simple calculation yields
\begin{align}
  \mmse( \gamma , \eps) &= \E\left\{ ( X_0 - \E\left\{ X_0 | X_0(\eps),Y_{0} \right\})^2  \right\}\\
  &= (1-\eps)\, \mmse(\gamma)\, .
\end{align}
\begin{propo}\label{prop:singleletter}
  For any $\lambda>0$, $\eps\in (0,1]$, let $\gamma_*(\lambda, \eps)$
  be the largest non-negative  solution of the equation:
 \begin{align}
      \gamma &= \lambda\,\big(1- (1-\eps)\mmse(\gamma)\big)\, . \label{eq:EpsFixedPoint}
\end{align}
Further, define $\Psi(\gamma, \lambda, \eps)$ by:
\begin{align}
\Psi(\gamma, \lambda, \eps) &=  
\frac{\lambda}{4}+\frac{\gamma^2}{4\lambda}-\frac{\gamma}{2}+\eps\,\log
2+(1-\eps)\,\Info(\gamma)\, . \label{eq:PsiEps}
\end{align}
  Then, we have
\begin{align}
    \lim_{n\to\infty}\frac{1}{n} I(\bX;\bX(\eps),\bY) &=\Psi\big(\gamma_*(\lambda,\eps),\lambda,\eps\big)\, .
\end{align}
\end{propo}

Using continuity in $\eps$, the last result implies directly 
a limit result for the mutual information under the Gaussian model,
which we single out since it is of independent interest.
\begin{thm}\label{thm:singleletterGauss}
  For any $\lambda>0$,  let $\gamma_*(\lambda)$
  be the largest non-negative solution of the equation:
 \begin{align}
      \gamma &= \lambda\,\big(1- \mmse(\gamma)\big)\, .
\end{align}
Further, define $\Psi(\gamma, \lambda)$ by:
\begin{align}
\Psi(\gamma, \lambda) &=  
\frac{\lambda}{4}+\frac{\gamma^2}{4\lambda}-\frac{\gamma}{2}+\Info(\gamma)\, .
\end{align}
  Then, we have
\begin{align}
    \lim_{n\to\infty}\frac{1}{n} I(\bX;\bY) &=\Psi\big(\gamma_*(\lambda),\lambda\big)\, .
\end{align}
\end{thm}

\subsection{Approximate Message Passing (AMP)}

To analyze the Gaussian model \myeqref{eq:gaussobs} we introduce 
an approximate message passing (AMP)  algorithm that computes estimates $\bx^t\in\reals^n$ at time $t$, which are functions
of the observations $\bY, \bX(\eps)$. This construction follows the 
general scheme of AMP algorithms developed in  \cite{DMM09,BM-MPCS-2011,javanmard2013state}. 
Given a sequence of functions
$f_t:\reals \times \left\{ -1, 0 , +1 \right\} \to\reals$, we set
$\bx^0 = 0$ and compute
\begin{align}
  \bx^{t+1} &= \frac{\bY(\lambda)}{\sqrt{n}} f_t(\bx^t , \bX(\eps)) -
  \ons_t f_{t-1}(\bx^{t-1}, \bX(\eps))\, , \label{eq:ampiter}\\
  \ons_t &= \frac{1}{n} \sum_{i=1}^n f_t'(\bx^t_i, \bX(\eps)_i)\, .
\end{align}
Above (and in the sequel) we extend the function $f_t$ to vectors by
applying it component-wise, i.e. $f_t(\bx^t, \bX(\eps)) = 
(f_t(x^t_1, X(\eps)_1), f_t(x^t_2, X(\eps)_2), \dots f_t(x^t_n,
X(\eps)_n))$. 

The AMP iteration above proceeds analogously to the usual power iteration to compute
principal eigenvectors, but has an additional memory term $-\ons_t f_{t-1}(\bx^{t-1})$. This
additional term changes the behavior of the iterates in an important way: unlike the usual power
iteration, there is an explicit distributional characterization of the iterates $\bx^t$
in the limit of large dimension. 
Namely, for each time $t$ we will show that, approximately $x^t_i$ is a scaled version of
the truth $X_i$ observed through Gaussian noise of a certain variance. We define the following
two-parameters recursion, with initialization $\mu_0=\sigma_0=0$, which will be referred to as  \emph{state evolution}:
\begin{align}
  \mu_{t+1} &= 
  \sqrt{\lambda}\, \E\left\{ X_0 f_t(\mu_t X_0 + \sigma_t Z_0 ,
    X_0(\eps) )\right\}\, ,\label{eq:stateevol1} \\
  \sigma^2_{t+1} &= \E\left\{ f_t(\mu_t X_0 + \sigma_t Z_0,  X_0(\eps) )^2
  \right\}\, ,\label{eq:stateevol2}
\end{align}
where expectation is with respect to the independent random variables
$X_0\sim\Unif(\{+1,-1\})$, $Z_0\sim\normal(0,1)$ and
$B_0\sim\Ber(1-\eps)$, setting $X_0(\eps) = B_0X_0$.

The following lemma makes this distributional characterization precise.
It follows from the more general result of \cite{javanmard2013state} and we provide a 
proof in Appendix \ref{app:addproofs}.
\begin{lemma}[State Evolution]
  \label{lem:stateevollem}
  Let $f_t:\reals\times\{-1, 0, 1\}\to\reals$ be a sequence of functions such that
  $f_t, f'_t$ are Lipschitz continuous in
  their first argument (where $f'_t$ denotes the derivative of $f_t$ with respect to the
  first argument).

  Let $\psi: \reals\times\{+,1-1\}\times\{+1,0,-1\} \to\reals$ 
be a test function such that $|\psi(x_1,s,r)-\psi(x_2,s,r)|\le
C(1+|x_1|+|x_2|)\, |x_1-x_2|$ for all $x_1,x_2,s,r$. 
Then the following limit holds almost surely for $(X_0,Z_0,X_0(\eps))$
random variables distributed as above
\begin{align}
    \lim_{n\to\infty}\frac{1}{n}\sum_{i=1}^n \psi(x^t_i, X_i,X(\eps)_i) &= 
    \E\left\{ \psi\big(\mu_t X_0 + \sigma_t Z_0, X_0,X_0(\eps)\big)
    \right\}\, ,
  \end{align}
\end{lemma}
Although the above holds for a relatively broad class of functions $f_t$, we are interested
in the AMP algorithm for specific functions $f_t$. Specifically, we
following sequence of  functions
\begin{align}
  f_t(y, s) &= \E\big\{ X_0 | \mu_t X_0 + \sigma_t Z_0 = y,\, X_0(\eps) = s
    \big\}\, .  \label{eq:optfunc}
\end{align}
It is easy to see that $f_t$ satisfy the requirement of Lemma \ref{lem:stateevollem}.
We will refer to this version of AMP as \emph{Bayes-optimal AMP}. 

Note that the definition  (\ref{eq:optfunc}) depends itself on $\mu_t$ and $\sigma_t$
defined through Eqs.~(\ref{eq:stateevol1}),
(\ref{eq:stateevol1}). This recursive definition is perfectly well
defined and yields 
\begin{align}
  \mu_{t+1} &= \sqrt{\lambda} \E\big\{ X_0 \E\{X_0| \mu_t X_0 + \sigma_t Z_0, X(\eps)_0\} \big\}, \\
  \sigma_{t+1}^2 &= \E\big\{ \E\{X_0|\mu_t X_0 + \sigma_t Z_0,
    X(\eps)_0\}^2 \big\}\, .
\end{align}
Using the fact that  $f_t(y, s)= \E\{ X_0 | \mu_t X_0 + \sigma_t Z_0 =
y,\, X_0(\eps) = s\}$
is the minimum mean square error estimator, we obtain
\begin{align}
  \mu_{t+1} &= \sqrt{\lambda} \sigma_{t+1}^2\, , \label{eq:simplstateevol1}\\
  \sigma_{t+1}^2 
  &= 1 - (1-\eps)\,\mmse( \lambda \sigma_t ^2)\,
  , \label{eq:simplstateevol2}
\end{align}
 where $\mmse(\,\cdot\,)$ is given explicitly by
 Eq.~(\ref{eq:mmseFormula}).

In other words, the state evolution recursion reduces to a simple
one-dimensional recursion that we can write in terms of the variable
$\gamma_t \equiv \lambda\sigma_t^2$. We obtain
\begin{align}
  \gamma_{t+1} &= \lambda \big(1 - (1-\eps)\,\mmse( \gamma_t)\big)\,
  ,\label{eq:simplstateevol3}\\
& \sigma^2_t = \frac{\gamma_t}{\lambda}\, ,\;\;\;\; \mu_t = \frac{\gamma_t}{\sqrt{\lambda}}
\end{align}

Our
proof strategy uses the AMP algorithm to construct estimates that \emph{bound from above}
the minimum error of estimating $\bX$ from observations  $\bY,\bX(\eps)$. However, in the limit of a large
number of iterations, we show that the gap between this upper bound
and the minimum estimation error
vanishes via an area theorem. 

More explicitly, we develop an upper bound on the matrix mean square
error first introduced in Eq.~(\ref{eq:MMSEdef}). We generalize this
in the obvious way to the Gaussian observation model:
\begin{align}
\MMSE(\lambda,\eps,n) &\equiv
\frac{1}{n^2}\,\E\Big\{\big\|\bX\bX^{\sT}-\E\{\bX\bX^{\sT}|\bX(\eps),\bY\}\big\|_F^2\Big\}\,. 
\end{align}
(Note that we adopt here a slightly different normalization with
respect to Eq.~(\ref{eq:MMSEdef}). This change is immaterial in the large
$n$ limit.)

We then use AMP to construct the sequence of estimators $\hbx^{t} =
f_{t-1}(\bx^{t-1}, \bX(\eps))$,
indexed by $t\in \{1,2,\dots\}$,  where $f_{t-1}$ is defined as in
\myeqref{eq:optfunc}. 
The matrix mean squared error  of this
estimators will be denoted by
\begin{align}
  \MSEAMP(t;\lambda, \eps, n)&\equiv \frac{1}{n^2}\,\E\Big\{
  \big\|\bX\bX^\sT-\hbx^t (\hbx^t) ^\sT\big\|^2 \Big\}\, .
\end{align}

We also define the limits
\begin{align}
\MSEAMP(t;\lambda, \eps) &\equiv \lim_{n\to\infty} \MSEAMP(t;\lambda,\eps)\, ,\\
\MSEAMP(\lambda,\eps) &= \lim_{t\to\infty}\mseAMP(t;\lambda, \eps)\, .
\end{align}
 In the course of the proof,
we will also see that these limits are well-defined, using the state
evolution  Lemma \ref{lem:stateevollem}

%
%**************
%

\subsection{Proof of Theorem \ref{thm:main} and Theorem  \ref{thm:singleletterGauss}}
\label{sec:ProofTheoremMain}

The proof is almost immediate given Propositions
\ref{prop:gausseq} and \ref{prop:singleletter}.
Firstly, note that, for any $\eps\in (0,1]$, 
\begin{align}
\left|\frac{1}{n} I(\bX;\bX(\eps),\bY) -\frac{1}{n} I(\bX;
  \bY)\right|\le\frac{1}{n} I(\bX;\bX(\eps),\bY) \le \eps\log 2
\end{align}
Since,  by Proposition \ref{prop:singleletter} $I(\bX;\bX(\eps),\bY)/n$ has a well-defined limit as $n\to\infty$, 
and $\eps>0$ is arbitrary, we have that:
\begin{align}
  \lim_{n\to\infty} \frac{1}{n} I(\bX;\bY) &= \lim_{\eps\to 0}
  \Psi(\gamma_*(\lambda, \eps), \lambda, \eps)\, . 
\end{align}

It is immediate to check that  $\Psi(\gamma, \lambda, \eps)$ is
continuous in $\eps\ge 0$, $\gamma\ge 0$
and $\Psi(\gamma, \lambda, 0) = \Psi(\gamma, \lambda)$ as defined
in Theorem \ref{thm:main}.
Furthermore, as $\eps\to 0$, the unique positive solution
$\gamma_*(\lambda, \eps)$ of Eq.~(\ref{eq:EpsFixedPoint})
converges to $\gamma_*(\lambda)$, the largest non-negative solution
to of Eq.~(\ref{eq:MainEquation}), which we copy here for the readers'
convenience:
\begin{align}
  \gamma &= \lambda(1-\mmse(\gamma)).
\end{align}
This follows from the smoothness and concavity of the function
$1-\mmse(\gamma)$ (see Lemma \ref{rem:unique}).
It follows that $\lim_{\eps\to 0}\Psi(\gamma_*(\lambda, \eps),
\lambda, \eps) =\Psi(\gamma_*(\lambda),
\lambda)$  and therefore
\begin{align}
  \lim_{n\to\infty} \frac{1}{n}I(\bX; \bY) &= \Psi(\gamma_*(\lambda),
  \lambda), 
\end{align}
This proves Theorem \ref{thm:singleletterGauss}.
Theorem \ref{thm:main}
follows by applying Proposition \ref{prop:gausseq}. 

%
%*********************
%

\section{Proof of Proposition \ref{prop:gausseq}}
\label{sec:Gausseq}

Given a collection $\bV = (V_{ij})_{i<j}$ of random variables defined on the same
probability space as $\bX$, and a non-negative real number $\lambda$,  
we define the following Hamiltonian and log-partition
function associated with it:
\begin{align}
  \cH(\bx, \bX, \bV, \lambda, n) &\equiv \sum_{i<j} V_{ij} (x_i x_j -
                                   X_i X_j) + \frac{\lambda}{n} x_ix_j
                                   X_iX_j \, ,\\
  \phi(X, V, \lambda, n) &\equiv \log \Big\{\sum_{x\in \hcube} \exp\big(\cH(\bx,
                           \bX, \bV, \lambda, n)\big)\Big\}\, .  
\end{align}
\begin{lemma}\label{lem:gaussentro}
  We have the identity:
  \begin{align}
    I(\bX; \bY) = n\log 2 + \frac{(n-1)\lambda}{2} - \E_{X, Z} \left\{
    \phi(\bX, \bZ\sqrt{\lambda/n}, \lambda, n) \right\}\, . 
\end{align}
\end{lemma}
\begin{proof}
  By definition:
  \begin{align}
    I(\bX; \bY) &= \E\left\{\log \frac{\d p_{\bY|\bX} (\bY | \bX)}{\d p_\bY (\bY(\lambda))} \right\}.
  \end{align}
  Since the two distributions $p_{\bY|\bX}$ and  $p_{\bY}$ are absolutely continuous with respect to each
  other, we can write the above simply in terms of the ratio of (Lebesgue) densities,
  and we obtain:
  \begin{align}
  I(X; Y) &= \E\log \left\{\frac{\exp \big(-\big\|\bY - \sqrt{\lambda/n}\bX\bX^\sT\big\|_F^2/4\big)}
  {\sum_{x\in\hcube} 2^{-n}\exp\big(-\big\|\bY - \sqrt{\lambda/n} \bx\bx^\sT\big\|_F^2/4\big) } \right\} \\
    &= n\log 2- \E_{\bX, \bZ} \log \left\{ \sum_{x\in \hcube} 
\exp\left( -\frac{1}{2}\sum_{i<j}\Big(Z_{ij} + \sqrt{\frac{\lambda}{n}}(X_iX_j - x_i x_j)\Big)^2 +\frac{1}{2} Z_{ij}^2)   \right) \right\} \\
    &= n\log 2-\E_{\bX, \bZ} \log \left\{ \sum_{x\in \hcube}
      \exp\left(\sum_{i<j} \sqrt{\frac{\lambda}{n}} Z_{ij}(x_ix_j -
      X_i X_j) - \frac{\lambda}{2n} (x_ix_j - X_i X_j)^2 \right)
      \right\}\, .\label{eq:gaussmutinfexpr}
  \end{align}
  We modify the final term as follows:
  \begin{align}
    \sum_{i<j} (x_i x_j - X_i X_j)^2 &= \sum_{i<j} 2 -2 x_i x_j X_i X_j \\
    &= n(n-1) - 2\sum_{i<j} x_i x_j X_i X_j.
  \end{align}
  Substituting this in \myeqref{eq:gaussmutinfexpr} we have
  \begin{align}
    I(\bX;\bY) &= n\log 2-\E_{\bX, \bZ} \log \left\{\sum_{x\in \hcube} %
    \exp\Big(\sum_{i<j} \sqrt{\frac{\lambda}{n}} Z_{ij}(x_ix_j - X_i X_j)%
    + \frac{\lambda}{n} x_ix_jX_iX_j \Big) \right\} + \frac{1}{2}\lambda(n-1),
  \end{align}
  as required. 
\end{proof}

\begin{lemma}\label{lem:sbmentro}
   Define the (random) Hamiltonian $\sbmH(x, \lambda, n)$ by:
   \begin{align}\label{eq:sbmentrodef}
     \sbmH(\bx, \bX, \bG, n) &\equiv \sum_{i<j}
                               \left\{G_{ij}\log\left( \frac{\op_n +
                               \Delta_n x_i x_j}{\op_n + \Delta_n
                               X_iX_j} \right) + (1-G_{ij})\log \left(
                               \frac{1-\op_n - \Delta_n x_i
                               x_j}{1-\op_n - \Delta_n X_iX_j}
                               \right)\right\}\, .
   \end{align}
Then we have that:
\begin{align}
     I(\bX;\bG) &=  n\log 2-\E_{\bX, \bG}\log  \left\{
                  \sum_{\bx\in\hcube} \exp(\sbmH(\bx,\bX, \bG, n))
                  \right\}\, .
   \end{align}
\end{lemma}
\begin{proof}
    This follows directly from the definition of mutual
    information: 
    \begin{align}\label{eq:sbmentrdef}
      I(\bX;\bG) &= \E_{\bX, \bG} \left\{ \frac{\d
                   p_{\bG|\bX}(\bG|\bX)}{\d p_{\bG}(\bG)} \right\}\, .
    \end{align}
    As in Lemma \ref{lem:gaussentro} we can write this in terms of densities as:
    \begin{align}
      \frac{\d p_{\bG|\bX}(\bG|\bX)}{\d p_\bG (\bG)} &= \frac{\prod_{i<j}(\op_n + \Delta_n X_i X_j)^{G_{ij}} (1- \op_n - \Delta_nX_i X_j)^{1-G_{ij}}}%
      {\sum_{x\in \hcube} 2^{-n}\prod_{i<j}(\op_n + \Delta_n x_i x_j)^{G_{ij}} (\op_n + \Delta_n X_i X_j)^{1-G_{ij}}  }
    \end{align}
    Substituting this in the mutual information formula \myeqref{eq:sbmentrdef} yields the lemma.
\end{proof}

Define the random variables $\btG = (\tG_{ij})_{i<j}$ as follows:
\begin{align}
  \tG_{ij} &\equiv \frac{\Delta_n}{\op_n(1-\op_n)}(G_{ij} - \op_n -\Delta_n X_iX_j).
\end{align}

  The following lemma shows that, to compute $I(\bX;\bG)$ it suffices to compute the 
  log-partition function with respect to the approximating Hamiltonian.
  \begin{lemma}
Assume that, as $n\to\infty$, $(i)$ $\lambda_n\to\lambda$ and $(ii)$  $n\op_n(1-\op_n)\to\infty$.
    Then, we have
    \begin{align}
      I(\bX;\bG) &=n\log 2 +\frac{(n-1)\lambda_n}{2} -  
      \E_{\bX, \btG}\left\{ \phi(\bX, \btG, \lambda_n, n) \right\}+
                   O\left(\frac{n\lambda^{3/2}}{\sqrt{n\op_n(1-\op_n)}}\right)\,
                   .
    \end{align}
    \label{lem:sbmapproxrel}
  \end{lemma}
  \begin{proof}
    We concentrate on the log-partition function for the hamiltonian
    $\sbmH(\bx, \bX, \bG, n)$. First, using the fact that $\log (c + dx) = \frac{1}{2}\log((c+d)(c-d))
    + \frac{x}{2}\log((c+d)/(c-d))$ when $x\in \{\pm 1\}$:
    \begin{align}\label{eq:sbmapprox1}
      \sbmH(\bx, \bX, \bG, n) &= (x_ix_j -X_iX_j)%
      \left( \frac{G_{ij}}{2} \log\left( \frac{1+\Delta_n/\op_n}{1-\Delta_n/\op_n}\right) %
      + \frac{(1-G_{ij})}{2} \log\left( \frac{1-\Delta_n/(1-\op_n)}{1+
                                \Delta_n/(1-\op_n)} \right) \right).
    \end{align}
    Now when $\max(\Delta_n/\op_n, \Delta_n/(1-\op_n))\le c_0$, for small enough $c_0$, we have
    by Taylor expansion the following approximation for $z\in [0, c_0]$:
    \begin{align}
      \abs{\frac{1}{2}\log\left( \frac{1+z}{1-z} \right) - z}
      &\le z^3,
    \end{align}
    which implies, by triangle inequality:
    \begin{align}\label{eq:sbmapprox2}
      \sbmH(\bx, \bX, \bG, n) &=\sum_{i<j}(x_ix_j - X_iX_j) 
      \left( \frac{\Delta_n G_{ij}}{\op_n} -
                                \frac{\Delta_n(1-G_{ij})}{1-\op_n}
                                \right) + \text{err}_n\, ,
\end{align}
where
\begin{align}
 \abs{\text{err}_n} &\le C\Delta_n^3 \left( \frac{\abs{\<\bx, \bG\bx\>} +\abs{\<\bX, \bG\bX\>} }{\op_n^3} %
      + \frac{\abs{\<\bx, (\one\one^\sT-\bG)\bx\>}+ \abs{\<\bX,
                      (\one\one^\sT-\bG)\bX\>}}{(1-\op_n)^3}\right).\label{eq:errndef}
    \end{align}

    We first simplify the RHS in \myeqref{eq:sbmapprox2}. Recalling the definition
    of $\tG_{ij}$: 
    \begin{align}
      \sum_{i<j} (x_ix_j - X_iX_j)\left( \frac{\Delta_n G_{ij}}{\op_n} - \frac{\Delta_n(1-G_{ij})}{(1-\op_n)} \right) 
      &= \sum_{i<j}(x_ix_j -X_iX_j) \left( \tG_{ij}+ \frac{ \Delta^2_{n} X_iX_j}{\op_n(1-\op_n)} \right) \\
      &= -\frac{(n-1)\lambda_n}{2} + \cH(\bx, \bX, \btG, \lambda_n, n). 
    \end{align}
    This implies that:
    \begin{align}
      \sbmH(\bx, \bX, \bG, n) &= -\frac{(n-1)\lambda_n}{2} + \cH(\bx, \bX, \btG, \lambda_n, n) + \text{err}_n, \label{eq:RHSsimple}
    \end{align}
    where $\text{err}_n$ satisfies \myeqref{eq:errndef}. 
    We now use the following remark, which is a simple
    application of Bernstein inequality (the proof is deferred
    to Appendix \ref{app:addproofs}).
    \begin{remark}\label{rem:berns}
    There exists a constant $C$ such that for every $n$ large enough:
      \begin{align}
	\P\left\{ \sup_{\bx\in\hcube} \abs{\<\bx, \bG\bx\>} \ge C n^2
        \op_n \right\} &\le \exp(-n^2\op_n/2)/2\, , \\
	\P\left\{ \sup_{\bx\in\hcube} \abs{\<\bx, (\one\one^\sT -
        \bG)\bx\>} \ge C n^2(1- \op_n) \right\} &\le
                                                  \exp(-n^2(1-\op_n)/2)/2\,
                                                  .
    \end{align}
    \end{remark}

    Using this Remark, the error bound \myeqref{eq:errndef} and
    \myeqref{eq:RHSsimple} in Lemma \ref{lem:sbmentro} yields 
    \begin{align}
      I(\bX;\bG)&= n\log2 + \frac{(n-1)\lambda_n}{2} - \E_{\bX, \bG}\left\{ \phi(\bX, \btG, \lambda_n, n) \right\}
      + O\left(\Delta_n^3\bigg( \frac{n^2}{\op_n^2} + \frac{n^2}{(1-\op_n)^2} \bigg)\right) \\
      &\quad+ O\left(\Delta_n^3\bigg( \frac{n^2 \exp(-n^2\op_n)}{\op_n^3} + \frac{n^2 \exp(-n^2(1-\op_n))}{(1-\op_n)^3}\bigg) \right)\nonumber\\
      &= n\log2 + \frac{(n-1)\lambda_n}{2} - \E_{X, G}\left\{ \phi(\bX, \btG, \lambda_n, n) \right\} 
      + O\left(\frac{n^2\Delta_n^3}{\op_n^2(1-\op_n)^2}\right).
  \end{align}
  Substituting $\Delta_n = (\lambda_n\op_n(1-\op_n)/n)^{1/2}$ gives the lemma.
  \end{proof}

  We now control the deviations that occur when replacing the variables $\tG_{ij}$ with Gaussian variables
  $Z_{ij}\sqrt{\lambda/n}$.
\begin{lemma}\label{lem:approx}
Assume that, as $n\to\infty$, $(i)$ $\lambda_n\to\lambda$ and $(ii)$
$n\op_n(1-\op_n)\to\infty$. Then we have:
    \begin{align}
      \E_{\bX, \btG} \left\{ \phi(\bX, \btG, \lambda_n, n) \right\} &= 
      \E_{\bX, \bZ}\left\{ \phi(\bX, \bZ\sqrt{\lambda/n}, \lambda, n)
                                                                      \right\} + O\left( \frac{n\lambda^{3/2}}{\sqrt{n\op_n(1-\op_n)}} + n\abs{\lambda_n-\lambda} \right).
    \end{align}
\end{lemma}
\begin{proof}
  This proof follows the Lindeberg strategy \cite{chatterjee2006generalization,korada2011applications}.
  We will show that:
  \begin{align}\label{eq:conditionalclaim}
    \E\Big\{ \phi(\bX, \btG, \lambda_n, n) \Big|\bX \Big\} &= 
    \E\Big\{ \phi(\bX, \bZ\sqrt{\lambda/n}, \lambda, n) \Big|\bX \Big\} 
    + O\left( \frac{n\lambda^{3/2}}{\sqrt{n\op_n(1-\op_n)}} +
                                                             n\abs{\lambda_n-\lambda} \right). 
  \end{align}
(with the $O(\cdots)$ term uniform in $\bX$).
  The claim then follows by taking expectations on both sides. 
  Note that, by construction:
  \begin{align}
    \E\big\{ \tG_{ij}\big|\bX \big\} &= \E \big\{
                                     Z_{ij}\sqrt{\lambda_n/n}\big|\bX
                                     \big\} = 0 \, ,\\
    \abs{ \E\big\{ \tG^2_{ij} \big|\bX \big\} - \E\big\{ Z_{ij}^2(\lambda_n/n) \big|X \big\}} 
    &= \abs{\frac{\Delta_n^2}{\op_n^2(1-\op_n)^2}(\op_n + \Delta_n
      X_iX_j)(1-\op_n -\Delta_nX_iX_j) - \frac{\lambda_n}{n}} \\
    &\le \frac{\lambda_n}{n}\left(
      \sqrt{\frac{\lambda_n}{n\op_n(1-\op_n)}}+\frac{\lambda_n}{n}\right)
      \, ,
\end{align}
and
\begin{align}
    \abs{\E\big\{ \tG_{ij}^3 \big|\bX \big\}} &=
                                                \bigg\lvert\frac{\Delta_n^3}{\op_n^3(1-\op_n)^3}\Big( (\op_n+\Delta_nX_iX_j) (1-\op_n-\Delta_nX_iX_j)^3
    \\
&\phantom{aaa}
  + (1-\op_n-\Delta_nX_iX_j)(-\op_n-\Delta_nX_iX_j)^3
  \Big)\bigg\rvert \nonumber\\
&\le
  \bigg\lvert\frac{\Delta_n^3}{\op_n^3(1-\op_n)^3}(\op_n+\Delta_nX_iX_j)
  (1-\op_n-\Delta_nX_iX_j)\bigg\rvert\\
&\le \frac{\Delta^3}{\op_n^2(1-\op_n)^2}+\frac{\Delta_n^4}{\op_n^3(1-\op_n)^3}+\frac{\Delta_n^5}{\op_n^3(1-\op_n)^3}\\
    &\le\frac{3\lambda_n^{3/2}}{n(n\op_n(1-\op_n))^{1/2}}\, .
  \end{align}
  We now derive the following estimates:
  \begin{align}
    \abs{\partial_{ij}^r \phi(\bX, \bz, \lambda_n, n)} &\le C, \text{ for } r = 1, 2, 3.
  \end{align}
  Here $\partial_{ij}^r \phi(\bX, \bz, \lambda_n, n)$ is the $r$-fold derivative of $\phi(\bX, \bz, \lambda_n, n)$ in 
  the entry $z_{ij}$ of the matrix $\bz$. To write explicitly
  the derivatives $\partial_{ij}\phi(\bX, \bz, \lambda_n, n)$ we introduce some notation. 
  For a function $f:\hcube\to\reals$, we write
  $\<f\>_\bz$ to denote its expectation with respect to the measure defined by the hamiltonian 
  $\cH(\bx, \bX, \bz, \lambda_n, n)$. Explicitly:
  \begin{align}
    \<f\>_\bz &\equiv e^{-\phi(\bX, \bz, \lambda,
                n)}\; \sum_{\bx\in\hcube} f(\bx) \exp\big(\cH(\bx, \bX, \bz,
                \lambda, n)\big) \, . 
  \end{align}
  Then the partial derivatives above can be expressed as
  \begin{align}
    \partial_{ij}\phi(\bX, \bz, \lambda, n) &=  \< x_ix_j - X_iX_j\>_z \\
    \partial^2_{ij}\phi(\bX, \bz, \lambda, n) &= \left( \<(x_i x_j - X_iX_j)^2\>_z - \<x_i x_j - X_iX_j\>_z^2 \right) \\
    \partial^3_{ij}\phi(\bX, \bz, \lambda, n) &= \big( \< (x_i x_j - X_iX_j)^3\>_z - 3 \<x_ix_j - X_iX_j\>_z\<(x_ix_j - X_iX_j)^2\>_z \\%
    &\quad+ 2 \<x_ix_j - X_iX_j\>_z^3  \big).
  \end{align}
  However since $\abs{x_ix_j - X_iX_j}\le 2$, we obtain:
  \begin{align}
    \abs{\partial^r_{ij}\phi(\bX, \bz, \lambda_n, n)} &\le C.
  \end{align}
  Applying Theorem 2 of \cite{korada2011applications} (stated below as
  Theorem \ref{thm:lindegen}) gives:
  \begin{align}
    \E\left\{ \phi(\bX, \btG, \lambda_n, n) \vert \bX \right\} &= 
    \E\left\{ \phi(\bX, \bZ\sqrt{\lambda/n}, \lambda_n, n) | \bX \right\} 
    +O\bigg(
                                                                 \frac{n\lambda_n^{3/2}}{\sqrt{n\op_n(1-\op_n)}}\bigg). \label{eq:linde1}
  \end{align}
  Further, we have:
  \begin{align}
    \big\lvert\partial_\lambda \E\left\{ \phi(\bX, \bZ\sqrt{\lambda'/n}), \lambda, n) \lvert \bX \right\}\big\rvert &=
    \frac{1}{n} \bigg\lvert\E\bigg\{ \<\sum_{i<j} x_ix_j X_iX_j\> \lvert \bX \bigg\}\bigg\rvert \\
    &\le \frac{n}{2}.
  \end{align}
  Here $\partial_{\lambda}$ denotes the derivative with respect to the variable $\lambda$.
  Thus,
  \begin{align}
    \E\left\{ \phi(\bX, \bZ\sqrt{\lambda/n}, \lambda_n, n) \lvert \bX \right\} 
    &= \E\left\{ \phi(\bX, \bZ\sqrt{\lambda/n}, \lambda, n) \lvert \bX \right\} 
    + O(n \abs{\lambda_n-\lambda}). \label{eq:linde2}
  \end{align}
  Combining Eqs. \eqref{eq:linde1}, \eqref{eq:linde2} gives 
  \myeqref{eq:conditionalclaim}, and the lemma follows by taking expectations on either
  side. 
\end{proof}
We state below the Lindeberg generalization theorem for
convenience:
\begin{thm}[Theorem 2 in \cite{korada2011applications}]\label{thm:lindegen}
  Suppose we are given two collections of random variables $(U_i)_{i\in [N]}$, $(V_{i})_{i\in[N]}$ with
  independent components and a function $f:\reals^N\to\reals$. Let $a_i = \abs{\E\{U_i\}- \E\{V_i\}}$
  and $b_i = \abs{\E\{U_i^2\} -\E\{V_i^2\}}$. Then:
  \begin{align}
    \abs{\E\left\{ f(U)  \right\} - \E\left\{ f(V) \right\}} &\le  
    \sum_{i=1}^N \bigg(  
    a_i \E\left\{ \abs{\partial_i f(U_1^{i-1}, 0, V_{i+1}^N)} \right\} 
    + \frac{b_i}{2} \E\left\{ \abs{\partial^2_if(U_1^{i-1}, 0, V_{i+1}^N)} \right\}  \nonumber\\
    &\quad+ \frac{1}{2} \E\int_{0}^{U_i}\big\lvert{\partial^3_i f(U_1^{i-1}, s, V_{i+1}^N)}\big\rvert(U_i - s)^2\d s \nonumber\\
    &\quad+ \frac{1}{2} \E\int_{0}^{V_i}\big\lvert{\partial^3_i f(U_1^{i-1}, s, V_{i+1}^N)}\big\rvert(V_i - s)^2\d s
    \bigg)\, .
  \end{align}
\end{thm}

With these in hand, we can now prove Proposition \ref{prop:gausseq}.
\begin{proof}[Proof of Proposition \ref{prop:gausseq}]

The proposition follows simply by combining the formulae for
$I(\bX; \bG), I(\bX;\bY)$ in Lemmas \ref{lem:gaussentro}, \ref{lem:sbmapproxrel}
with the approximation guarantee of Lemma \ref{lem:approx}.
\end{proof}
%
%******************************************
%
\section{Proof of Proposition \ref{prop:singleletter}} 
\label{sec:singleletterproof}

Throughout this section we will write $\bY(\lambda)$ whenever we want
to emphasize the dependence of the law of $\bY$ on the signal to noise 
parameter $\lambda$.

The proof of Proposition \ref{prop:singleletter} follows essentially from
a few preliminary lemmas.

\subsection{Auxiliary lemmas}

We begin with some properties of the fixed
point equation (\ref{eq:EpsFixedPoint}). The proof of this lemma can
be found in Appendix \ref{sec:ProofUnique}.
\begin{lemma}\label{rem:unique}
For any $\eps\in [0,1]$, the following properties 
hold for the function $\gamma\mapsto (1-\mmse(\gamma,\eps)) =
  (1-(1-\eps)\mmse(\gamma))$:
\begin{enumerate}
\item[$(a)$]  It is continuous, monotone increasing and concave in
  $\gamma\in\reals_{\ge 0}$. 
\item[$(b)$] It satisfies the following limit behaviors
  \begin{align}
    1-\mmse(0,\eps) &= \eps \, ,\\
    \lim_{\gamma\to\infty}\big[1- \mmse(\gamma,\eps)\big] &= 1\, . 
  \end{align}
\end{enumerate}
As a consequence we have the following for all $\eps\in(0,1]$:
\begin{enumerate}
\item[$(c)$] A non-negative solution
$\gamma_*(\lambda,\eps)$ of Eq.~(\ref{eq:EpsFixedPoint}) exists and is
unique for all $\eps>0$.
\item[$(d)$] For any $\eps>0$, the function $\lambda\mapsto
  \gamma_*(\lambda,\eps)$ is differentiable in $\lambda$.
\item[$(e)$] Let $\{\gamma_t\}_{t\ge 0}$ be defined recursively by
Eq.~(\ref{eq:simplstateevol3}), with initialization $\gamma_0=0$. Then
$\lim_{t\to\infty}\gamma_t = \gamma_*(\lambda,\eps)$.
\end{enumerate}
\end{lemma}

We then compute the value of $\Psi(\gamma_*(\lambda, \eps), \lambda, \eps)$
at $\lambda = \infty$ and $\lambda = 0$.
\begin{lemma}\label{lem:limpoints}
  For any $\eps>0$:
  \begin{align}
    \lim_{\lambda\to 0} \Psi(\gamma_*(\lambda, \eps), \lambda, \eps)
    &= \eps \log 2 \, ,\label{eq:PsiLamblda0}\\
    \lim_{\lambda\to\infty} \Psi(\gamma_*(\lambda, \eps), \lambda,
    \eps)&= \log 2\, .  \label{eq:PsiLambldaInfty}
  \end{align}
\end{lemma}
\begin{proof}
Recall the definition of $\mmse(\gamma)$,
cf. Eq.~(\ref{eq:mmseDef}). Upper bounding $\mmse(\gamma)$ by the
minimum error obtained by linear estimator yields, for any $\gamma\ge
0$, $0\le \mmse(\gamma) \le 1/(1+\gamma)$.
Substituting these bounds in Eq.~(\ref{eq:EpsFixedPoint}), we obtain
\begin{align}
\max(0,\gamma_{\rm LB}(\lambda,\eps))&\le\gamma_*(\lambda,\eps)\le  \lambda\,
                               ,\label{eq:BoundGamma}\\
\gamma_{\rm LB}(\lambda,\eps)& =
                               \frac{1}{2}\Big[\lambda-1+\sqrt{(\lambda-1)^2+4\lambda\eps}\Big]
                               = \lambda-(1-\eps)+O(\lambda^{-1})\, ,
\end{align}
where the last expansion holds as $\lambda\to\infty$.

Let us now consider the limit $\lambda\to 0$, cf. claim
(\ref{eq:PsiLamblda0}). Considering Eq.~(\ref{eq:PsiEps}), and using
$0\le \gamma_*(\lambda,\eps)\le \lambda$, we have
$(\lambda/4)+(\gamma_*(\lambda,\eps)^2/(4\lambda))-(\gamma/2) =
O(\lambda)\to 0$. Further from the definition (\ref{eq:InfoDef}) it
follows\footnote{This follows either by general information theoretic
  arguments, or else using dominated convergence in
  Eq.~(\ref{eq:InfoFormula}).} that $\lim_{\gamma\to 0}\Info(\gamma) =
0$ thus yielding Eq.~(\ref{eq:PsiLamblda0}).

Consider next the $\lambda\to\infty$ limit
of Eq.~(\ref{eq:PsiLambldaInfty}). In this limit
Eq.~(\ref{eq:BoundGamma})
implies $\gamma_*(\lambda,\eps) = \lambda +\delta_*\to \infty$, where
$\delta_*= \delta_*(\lambda,\eps) = O(1)$. Hence
$\lim_{\lambda\to\infty}\Info(\gamma_*(\lambda,\eps))
=\lim_{\gamma\to\infty}\Info(\gamma) =\log 2$ (this follows again from
the definition of $\Info(\gamma)$). Further
\begin{align}
\frac{\lambda}{4}+\frac{\gamma_*^2}{4\lambda}-\frac{\gamma_*}{2} =
  \frac{\delta_*^2}{4\lambda}= O(\lambda^{-1}) \, .
\end{align}
Substituting in Eq.~(\ref{eq:PsiEps}) we obtain the desired claim.
\end{proof}

The next lemma characterizes the limiting matrix mean squared
error of the AMP estimates. 
\begin{lemma} \label{lem:fixedpt}
Let $\{\gamma_t\}_{t\ge 0}$ be defined recursively by
Eq.~(\ref{eq:simplstateevol3}) with initialization $\gamma_0 = 0$,
and recall that $\gamma_*(\lambda,\eps)$ denotes the unique
non-negative solution of \myeqref{eq:EpsFixedPoint}.

Then the following limits hold for the AMP mean square error
\begin{align}
\MSEAMP(t;\lambda, \eps) &\equiv \lim_{n\to\infty}\MSEAMP(t;\lambda, \eps, n)
= 1-\frac{\gamma_t^2}{\lambda^2}\, ,\label{eq:MSEAMP_t}\\
\MSEAMP(\lambda, \eps) &\equiv \lim_{t\to\infty}\MSEAMP(t;\lambda, \eps)
= 1-\frac{\gamma_*^2}{\lambda^2}\, .\label{eq:MSEAMP_tinfty}
\end{align}
\end{lemma}
\begin{proof}
Note that Eq.~(\ref{eq:MSEAMP_tinfty}) follows from
Eq.~(\ref{eq:MSEAMP_t}) using Lemma \ref{rem:unique}, point $(d)$. We will
therefore focus on proving  Eq.~(\ref{eq:MSEAMP_t}).

First notice that:
\begin{align}
    \MSEAMP(t;\lambda, \eps, n) &= \frac{1}{n^2} \E\left\{ \norm{\bX \bX^\sT- \hbx^t (\hbx^t)^\sT }_F^2 \right\} \\
    &= \E\left\{  \frac{\norm{\bX}^4}{n^2}+
      \frac{\norm{\hbx^t}^2}{n^2}  - 2\frac{\<\bX, \hbx^t\>^2}{n^2}
    \right\}\, .
\end{align}
Since $\|\bX\|_2^2 = n $, the first term evaluates to 1. We use Lemma \ref{lem:stateevollem}
  to deal with the final two terms. Consider the last term $\<\bX, \hbx^t\>^2/n^2$. Using
  Lemma \ref{lem:stateevollem} with the $\psi(x, s,r) = f_{t-1}(x, r)
  s$ we have, almost surely
  \begin{align}
    \lim_{n\to \infty} \frac{1}{n} \sum_{i=1}^n
    f_{t-1}(x^{t-1}_i,X(\eps)_i) X_i &= 
\E\left\{ X_0 \E\left\{ X_0 | \mu_{t-1} X_0 + \sigma_{t-1} Z_0 ,
    X(\eps)_0\right\} \right\} \\ 
    &= \frac{\mu_{t}}{\sqrt{\lambda}} = \frac{\gamma_t}{\lambda}\, .
\end{align}
Note also that $\abs{X_i}$ and $\abs{\hx^t_i}$ are bounded by 1, hence so is
$\<\bX, \hbx^t_i\>/n$. It follows from the
bounded convergence theorem that
\begin{align}
  \lim_{n\to\infty} \E\left\{ \frac{\<\bX, \hbx^{t}\>^2}{n^2} \right\}
  &= \frac{\gamma_t^2}{\lambda^2}\, .
\end{align}
In a similar manner, we have that $\lim_{n\to\infty} \E\left\{
  \|\hbx^t\|_2^4/n^2 \right\} = \gamma_{t}^{2}/\lambda^2$,
whence the thesis follows.
\end{proof}

\begin{lemma} \label{lem:PsiIntegral}
  For every $\lambda\ge 0$ and $\eps >0$:
  \begin{align}
    \Psi(\gamma_*(\lambda,\eps), \lambda, \eps)&= \eps\log 2 +
    \frac{1}{4}\int_{0}^{\lambda}\MSEAMP(\tilde{\lambda},\eps)\,
    \mathrm{d}\tilde{\lambda}\, .
  \end{align}
\end{lemma}
\begin{proof}
By differentiating Eq.~(\ref{eq:PsiEps}) we  obtain (recall
$\frac{\d\Info(\gamma)}{\d \gamma} = (1/2)\mmse(\gamma)$):
\begin{align}
  \diffp{\Psi(\gamma, \lambda, \eps)}{\gamma}\bigg\rvert_{\gamma =
    \gamma_*} &=
  \frac{\gamma_*}{2\lambda}-\frac{1}{2}+\frac{1}{2}(1-\eps)\mmse(\gamma_*)
  = 0\, , \\
  \diffp{\Psi(\gamma, \lambda, \eps)}{\lambda}\bigg\rvert_{\gamma =
    \gamma_*} &= 
\frac{1}{4}\left( 1-\frac{\gamma_*^2}{\lambda^2}
\right)\, .
\end{align}
It follows from the uniqueness and differentiability of
$\gamma_*(\lambda,\eps)$ (cf. Lemma \ref{rem:unique}) that
$\lambda \mapsto \Psi(\gamma_*(\lambda,\eps), \lambda, \eps)$ is
differentiable for any fixed $\eps>0$, with derivative
\begin{align}
\frac{\d\Psi}{\d\lambda}(\gamma_*(\lambda,\eps), \lambda, \eps)&= 
\frac{1}{4}\left( 1-\frac{\gamma_*^2}{\lambda^2}
\right) \, .
\end{align}
The lemma follows from the fundamental theorem of calculus using 
Lemma \ref{lem:limpoints} for $\lambda=0$, and Lemma
\ref{lem:fixedpt}, cf. Eq.~(\ref{eq:MSEAMP_tinfty}).
\end{proof}

\subsection{Proof of Proposition \ref{prop:singleletter}}

We are now in a position to prove Proposition \ref{prop:singleletter}.
We start from a simple remark, proved in Appendix 
\begin{remark}\label{rem:limitentro}
We have
\begin{align}
\big|I(\bX;\bY(\lambda), \bX(\eps))-I(\bX\bX^{\sT};\bY(\lambda),
\bX(\eps))\big| \le \log 2\, . \label{eq:IXX}
\end{align}
Further the asymptotic mutual information satisfies
\begin{align}
 \lim_{n\to\infty}\frac{1}{n}I(\bX\bX^\sT;  \bX(\eps),\bY(0)) &= \eps\log 2\, ,\label{eq:IXX_lambda0}\\
    \lim_{\lambda\to\infty} \liminf_{n\to\infty} \frac{1}{n}
    I(\bX\bX^\sT;\bX(\eps),\bY(\lambda)) &= \log 2\, . \label{eq:IXX_lambdainfty}
  \end{align}
\end{remark}
We defer the proof of these facts to Appendix \ref{sec:limEntro}.

Applying the (conditional) I-MMSE identity of \cite{guo2005mutual} we have
\begin{align}
\frac{1}{n}  \diffp{I(\bX\bX^\sT ; \bY(\lambda), \bX(\eps))}{\lambda} &= 
  \frac{1}{2n^2}\sum_{i<j} \E\left\{ (X_iX_j - \E\{X_iX_j|\bX(\eps), \bY(\lambda)\})^2 \right\} \\
  &= \frac{1}{4} \MMSE(\lambda,\eps,n) \label{eq:IMMSE}\\
 &\le  \frac{1}{4n^2} \E\left\{ \big\|\bX\bX^{\sT}- \hbx^t (\hbx^t)^{\sT}\big\|_F^2 \right\}\\
  &= \frac{1}{4}\MSEAMP(t;\lambda, \eps, n)\, . \label{eq:IMMSE2}
\end{align}

We therefore  have
\begin{align}
(1-\eps)\log 2 & \stackrel{(a)}{=}
\lim_{\lambda\to\infty}\liminf_{n\to\infty}\frac{1}{n}\big[I(\bX;\bX(\eps),\bY(\lambda))-
I(\bX;\bX(\eps),\bY(0))\big] \\
&\stackrel{(b)}{=}
\lim_{\lambda\to\infty}\liminf_{n\to\infty}
\frac{1}{4} \int_{0}^{\lambda}\MMSE(\lambda',\eps,n) \, \d\lambda'\\
&\stackrel{(c)}{\le}
\lim_{\lambda\to\infty}\lim_{t\to \infty}\limsup_{n\to\infty}
\frac{1}{4} \int_{0}^{\lambda}\MSEAMP(t;\lambda',\eps,n) \,
\d\lambda'\\
&\stackrel{(d)}{=}
\lim_{\lambda\to\infty}\lim_{t\to \infty}
\frac{1}{4} \int_{0}^{\lambda}\MSEAMP(t;\lambda',\eps) \,
\d\lambda'
\end{align}
where $(a)$ follows from Remark \ref{rem:limitentro}, $(b)$ from
Eq.~(\ref{eq:IXX}) and (\ref{eq:IMMSE}), $(c)$ from
(\ref{eq:IMMSE2}), and $(d)$ from bounded convergence. 
Continuing from the previous chain we get
\begin{align}
(\;\cdots\;) &\stackrel{(e)}{=} \lim_{\lambda\to\infty}
\frac{1}{4} \int_{0}^{\lambda}\MSEAMP(\lambda',\eps) \,
\d\lambda' \\
& \stackrel{(f)}{=} \lim_{\lambda\to\infty} \big[\Psi(\gamma_*(\lambda,\eps), \lambda,
\eps)-\Psi(\gamma_*(0,\eps), 0, \eps)\big]\\
& \stackrel{(g)}{=} (1-\eps)\log 2\, ,
\end{align}
where $(e)$ follows from Lemma \ref{lem:fixedpt}, $(f)$ from Lemma
\ref{lem:PsiIntegral}, and $(g)$ from Lemma \ref{lem:limpoints}.

We therefore have a chain of equalities, whence the inequality $(c)$
must hold with equality. Since $\MMSE(\lambda,\eps,n) \le
\MSEAMP(t;\lambda,\eps,n)$
for any $\lambda$, this implies
\begin{align}
\MSEAMP(\lambda,\eps) =\lim_{t\to\infty}\lim_{n\to\infty}\MSEAMP(t;\lambda,\eps,n) =
\lim_{n\to\infty}\MMSE(\lambda,\eps,n)\, 
\end{align}
for almost every $\lambda$. The conclusion follows for every $\lambda$
by the monotonicity of  $\lambda\mapsto\MMSE(\lambda,\eps,n)$, and the
continuity of $\MSEAMP(\lambda,\eps)$.

Using again Remark  \ref{rem:limitentro},  and the last display, we
get that the following limit exists
\begin{align}
\lim_{n\to\infty} \frac{1}{n}\, I(\bX ; \bX(\eps),\bY(\lambda) ) 
&= \lim_{n\to\infty} \frac{1}{n} I(\bX\bX^\sT ; X(\eps)
,Y(\lambda)) \\
&\le \eps\log 2 +\lim_{n\to\infty}\frac{1}{4}\int_{0}^{\lambda}\MMSE(\lambda',\eps,
n)\, \d\lambda'\\
&\le \eps\log 2 +\frac{1}{4}\int_{0}^{\lambda}\MSEAMP(\lambda',\eps)\,
\d\lambda'\\
& = \Psi(\gamma_*(\lambda, \eps), \lambda, \eps)\, ,
\end{align}
where we used Lemma \ref{lem:PsiIntegral} in the last step. This
concludes the proof.
%
%******************************************
%
\section{Proof of Theorem \ref{thm:MainEstimation}}
\label{sec:ProofEstimation}

\subsection{A general differentiation formula}

In this section we recall a general formula
to compute the derivative of the conditional entropy $H(\bX|\bG)$ with
respect to noise parameters. The formula was proved in \cite{measson2004life}
and \cite[Lemma 2]{measson2009generalized}. We restate it in the
present context and present a self-contained proof for the reader's
convenience.

We consider the following setting. For $n$ an integer, denote by $\Pair$
the set of unordered pairs in $[n]$ (in particular $|\Pair| =
\binom{n}{2})$). We will use $e,e_1,e_2,\dots$ to denote elements of $\Pair$. 
For for each $e = (i,j)$ we are given a one-parameter family of
discrete noisy channels indexed by $\theta\in J$ (with $J = (a,b)$ a
non-empty interval), with input alphabet $\{+1,-1\}$ and finite output
alphabet $\cY$. Concretely, for any $e$, we have a transition
probability 
\begin{align}
\{p_{e,\theta}(y|x)\}_{x\in\{+1,-1\},y\in\cY}\, ,
\end{align}
which is differentiable in $\theta$. We shall omit the subscript
$\theta$ since it will be clear from the context.
 
We then consider  $\bX= (X_1,X_2,\dots,X_n)$ 
a random vector in $\{+1,-1\}^n$, and $\bY = (Y_{ij})_{(i,j)\in\Pair}$
a set of observations in $\cY^{\Pair}$ that are conditionally
independent given $\bX$. Further $Y_{ij}$ is the noisy observation of
$X_iX_j$ through the channel $p_{ij}(\,\cdot\,|\,\cdot\,)$. In
formulae, the joint probability density function of $\bX$ and $\bY$ is
\begin{align}
p_{\bX,\bY}(\bx,\by) = p_{\bX}(\bx) \prod_{(i,j)\in
  \Pair}p_{ij}(y_{ij}|x_ix_j)\, .
\end{align}
This obviously include the two-groups stochastic block model 
as a special case, if we take $p_{\bX}(\,\cdot\,)$  to be the uniform
distribution over $\{+1,-1\}^n$, and output alphabet $\cY = \{0,1\}$.
In that case $\bY = \bG$ is just the adjacency matrix of the graph.

In the following we write $\bY_{-e}= (Y_{e'})_{e'\in\Pair\setminus e}$
for the set of observations excluded $e$, and $X_e = X_{i}X_j$ for $e=(i,j)$. 
\begin{lemma}[\cite{measson2009generalized}]\label{lemma:Differentiation}
With the above notation, we have:
\begin{align}
\frac{\d H(\bX|\bY)}{\d \theta} = \sum_{e\in\Pair}\sum_{x_e,y_e}\frac{\d p_e(y_e|x_e)}{\d\theta}
\E\Big\{p_{X_e|\bY_{-e}}(x_e|\bY_{-e})\log
\Big[\sum_{x'_e}\frac{p_e(y_e|x'_e)}{p_e(y_e|x_e)}p_{X_e|\bY_{-e}}(x'_e|\bY_{-e})\Big]\Big\}\,
.
\label{eq:BigDerivative}
\end{align}
\end{lemma}
\begin{proof}
Fix $e\in\Pair$. By linearity of differentiation, it is sufficient to
prove the claim when only $p_e(\,\cdot\,|\,\cdot\,)$ depends on $\theta$.

Writing $H(\bX,Y_e|\bY_{-e})$ by chain rule in two
alternative ways we get
\begin{align}
H(\bX|\bY) + H(Y_e|\bY_{-e}) &= H(\bX|\bY_{-e}) +H(Y_e|\bX,\bY_{-e})\\ 
& = H(\bX|\bY_{-e}) +H(Y_e|X_e)\, ,
\end{align}
where in the last identity we used the conditional independence of
$Y_e$ from $\bX$, $\bY_{-e}$, given $X_e$. Differentiating with
respect to $\theta$, and using the  fact that $H(\bX|\bY_{-e})$ is
independent of $p_e(\,\cdot\,|\,\cdot\,)$, we get
\begin{align}
\frac{\d H(\bX|\bY)}{\d \theta} = \frac{\d H(Y_e|X_e)}{\d \theta}
-\frac{\d H(Y_e|\bY_{-e})}{\d \theta} \, .\label{eq:DerivativeFirst}
\end{align}
Consider the first term. Singling out the dependence of $H(Y_e|X_e)$
on $p_e$ we get
\begin{align}
\frac{\d H(Y_e|X_e)}{\d \theta} &=
-\frac{\d\phantom{a}}{\d\theta}\sum_{y_e}\E\big\{p_e(y_e|X_e)\log
p_e(y_e|X_e)\big\}\\
& =
-\sum_{y_e}\E\big\{\frac{\d p_e(y_e|X_e)}{\d\theta}\log
p_e(y_e|X_e)\big\}\\
& =
-\sum_{x_e,y_e}\frac{\d p_e(y_e|x_e)}{\d\theta}\, p_{X_e}(x_e)\, \log
p_e(y_e|x_e)\, .\label{eq:Derivative2}
\end{align}
In the second line we used the fact that the distribution of $X_e$ is
independent of $\theta$, and the normalization condition
$\sum_{y_e}\frac{\d p_e(y_e|x_e)}{\d\theta}=0$.

We follow the same steps for the second term
(\ref{eq:DerivativeFirst}):
\begin{align}
&\frac{\d H(Y_e|\bY_{-e})}{\d \theta} =-\frac{\d\phantom{a}}{\d\theta}\sum_{y_e}
\E\Big\{\Big[\sum_{x_e}p_{X_e,Y_e|\bY_{-e}}(x_e,y_e|\bY_{-e})\Big]\log
\Big[\sum_{x'_e}p_{X_e,Y_e|\bY_{-e}}(x'_e,y_e|\bY_{-e})\Big]\Big\}\\
&\phantom{AA}=-\frac{\d\phantom{a}}{\d\theta}\sum_{y_e}
\E\Big\{\Big[\sum_{x_e}p_e(y_e|x_e)p_{X_e|\bY_{-e}}(x_e|\bY_{-e})\Big]\log
\Big[\sum_{x'_e}p_e(y_e|x'_e)p_{X_e|\bY_{-e}}(x'_e|\bY_{-e})\Big]\Big\}\\
&\phantom{AA}=-\sum_{x_e,y_e}\frac{\d p_e(y_e|x_e)}{\d\theta}
\E\Big\{p_{X_e|\bY_{-e}}(x_e|\bY_{-e})\log
\Big[\sum_{x'_e}p_e(y_e|x'_e)p_{X_e|\bY_{-e}}(x'_e|\bY_{-e})\Big]\Big\}\, .\label{eq:Derivative3}
\end{align}
Taking the difference of Eq.~(\ref{eq:Derivative2}) and
Eq.~(\ref{eq:Derivative3}) we obtain the desired formula.
\end{proof}

\subsection{Application to the stochastic block model}

We next apply the general differentiation Lemma
\ref{lemma:Differentiation} to the stochastic block model.
As mentioned above, this fits the framework in the previous section,
by setting $\bY$ be the adjacency matrix of the graph $\bG$, and
taking $p_{\bX}$ to be the uniform distribution over $\{+1,-1\}^n$.
For the sake of convenience, we will encode this as $Y_e = 2G_e$.
In other words $\cY = \{+1,-1\}$ and $Y_e = +1$ (respectively $=-1$) 
encodes the fact that edge $e$ is present (respectively, absent).
We then have the following channel model for all $e\in\Pair$:
\begin{align}
p_e(+|+) &= p_n\, , &p_e(+|-) = q_n\, ,\\
p_e(-|+) &= 1-p_n\, , &p_e(-|-) = 1-q_n\, .
\end{align}
We parametrize these probability kernels by a common parameter
$\theta\in\reals_{\ge 0}$ by letting
\begin{align}
p_n = \op_n +\sqrt{\frac{\op_n(1-\op_n)}{n}\,\theta}\, ,\;\;\;\;\;\;\;\;\;\;
q_n = \op_n -\sqrt{\frac{\op_n(1-\op_n)}{n}\,\theta}\, .\label{eq:PnQnParam}
\end{align}
We will be eventually interested in setting $\theta=\lambda_n$ to make
contact with the setting of Theorem \ref{thm:MainEstimation}.
\begin{lemma}\label{lemma:DifferentiationSBM}
Let $I(\bX;\bG)$ be the mutual information of the two-groups
stochastic block models with parameters $p_n = p_n(\theta)$ and 
$q_n= q_n(\theta)$ given by Eq.~(\ref{eq:PnQnParam}). Then there
exists a numerical constant $C$ such that the following happens.

For any $\theta_{\rm max}>0$ there exists $n_0(\theta_{\rm max})$ such
that, if $n\ge n_0(\theta_{\rm max})$ then for all $\theta\in
[0,\theta_{\rm max}]$,
\begin{align}
\left|\frac{1}{n}\,\frac{\d I(\bX;\bG)}{\d \theta} -\frac{1}{4}\MMSE_n(\theta)\right|\le
  C\left(\sqrt{\frac{\theta}{n\op_n(1-\op_n)}} \vee \frac{1}{n}\right)\, .
\end{align}
\end{lemma}
\begin{proof}
We let $Y_e = 2G_2-1$ and apply Lemma \ref{lemma:Differentiation}.
Simple calculus yields
\begin{align}
\frac{\d p_e(y_e|x_e)}{\d \theta} =
\frac{1}{2}\sqrt{\frac{\op_n(1-\op_n)}{n\theta}}\; x_ey_e\, .
\end{align}

From Eq.~(\ref{eq:BigDerivative}), letting $\hx_e(\bY_{-e}) \equiv \E\{X_e|\bY_{-e}\}$,
\begin{align}
\frac{\d H(\bX|\bY)}{\d \theta} =& \frac{1}{2}\sqrt{\frac{\op_n(1-\op_n)}{n\theta}}\sum_{e\in\Pair}\sum_{x_e,y_e}x_ey_e
\E\Big\{p_{X_e|\bY_{-e}}(x_e|\bY_{-e})\log
\Big[\sum_{x'_e}p_e(y_e|x'_e) p_{X_e|\bY_{-e}}(x'_e|\bY_{-e})\Big]\Big\}\nonumber\\
&-\frac{1}{2}\sqrt{\frac{\op_n(1-\op_n)}{n\theta}}\sum_{e\in\Pair}\sum_{x_e,y_e}x_ey_e
p_{X_e}(x_e)\log p_e(y_e|x_e)\\
=& \frac{1}{2}\sqrt{\frac{\op_n(1-\op_n)}{n\theta}}\sum_{e\in\Pair}
\E\left\{\hx_e(\bY_{-e})\log
\left[\frac{\sum_{x'_e}p_e(+1|x'_e)
    p_{X_e|\bY_{-e}}(x'_e|\bY_{-e})}{\sum_{x'_e}p_e(-1|x'_e)
    p_{X_e|\bY_{-e}}(x'_e|\bY_{-e})}\right]\right\}\nonumber\\
&
-\frac{1}{4}\sqrt{\frac{\op_n(1-\op_n)}{n\theta}}\binom{n}{2}\log\Big\{\frac{p_e(+|+)p_e(-|-)}{p_e(+|-)p_e(-|+)}\Big\}\, . \label{eq:DH}
\end{align}
Notice that, letting $\Delta_n = \sqrt{\op_n(1-\op_n)\theta/n}$, 
\begin{align}
\frac{\sum_{x'_e}p_e(+1|x'_e)
    p_{X_e|\bY_{-e}}(x'_e|\bY_{-e})}{\sum_{x'_e}p_e(-1|x'_e)
    p_{X_e|\bY_{-e}}(x'_e|\bY_{-e})} &= \frac{(p_n+q_n)
                                       +(p_n-q_n)\hx_e(\bY_{-e})}{(2-p_n-q_n)
                                       -(p_n-q_n)\hx_e(\bY_{-e})}\\
&=\frac{\op_n}{1-\op_n}\, \frac{1+(\Delta_n/\op_n)\hx_{e}(\bY_{-e})}
{1-(\Delta_n/(1-\op_n))\hx_{e}(\bY_{-e})}\, ,\\
\frac{p_e(+|+)p_e(-|-)}{p_e(+|-)p_e(-|+)}& =
                                            \frac{(1+(\Delta_n/\op_n))(1+\Delta_n/(1-\op_n))}{(1-(\Delta_n/\op_n))(1-\Delta_n/(1-\op_n))}\,.
\end{align}
Since have $|\Delta_n/\op_n|, |\Delta_n/(1-\op_n)|\le
\sqrt{\theta/[\op_n(1-\op_n)n]}\to 0$, and $\hx_e(\bY_{-e})|\le 1$, we
obtain the following bounds
by Taylor expansion
\begin{align}
\left|\log\left[\frac{\sum_{x'_e}p_e(+1|x'_e)
    p_{X_e|\bY_{-e}}(x'_e|\bY_{-e})}{\sum_{x'_e}p_e(-1|x'_e)
    p_{X_e|\bY_{-e}}(x'_e|\bY_{-e})} \right]-B_0-
  \frac{\Delta_n}{\op_n(1-\op_n)}\,\hx_e(\bY_{-e})\right|\le C\,
  \frac{\theta}{\op_n(1-\op_n)n}\, ,\\
\left|\log\Big[\frac{p_e(+|+)p_e(-|-)}{p_e(+|-)p_e(-|+)}\Big]- \frac{2\Delta_n}{\op_n(1-\op_n)}\right|\le C\,
  \frac{\theta}{\op_n(1-\op_n)n}\, ,
\end{align}
where $B_0\equiv\log(\op_n/(1-\op_n))$ and $C$ will denote a numerical constant that will change from line
to line in the following. Such bounds hold for all $\theta\in
[0,\theta_{\rm max}]$ provided $n\ge n_0(\theta_{\rm max})$.

Substituting these bounds in Eq.~(\ref{eq:DH}) and using
$\E\{\hx_e(\bY_{-e})\} = \E\{X_e\} = 0$, after some manipulations, we get
\begin{align}
\left|\frac{1}{n-1}\,\frac{\d H(\bX|\bY)}{\d \theta} +
  \frac{1}{4}\binom{n}{2}^{-1}\sum_{e\in\Pair}
\Big(1-\E\{\hx_e(\bY_{-e})^2\}\Big)\right|\le
  C\sqrt{\frac{\theta}{n\op_n(1-\op_n)}}\, . \label{eq:DHFinal}
\end{align}

We now define (with a slight overloading of notation)
$\hx_e(\bY)\equiv \E\{X_e|\bY\}$, and relate $\hx_e(\bY_{-e})= \E\{X_e|\bY_{-e}\}$  to the overall
conditional expectation $\hx_e(\bY)$. By Bayes formula
we have
\begin{align}
p_{X_e|\bY}(x_e|\bY)= \frac{p_e(x_e|Y_e)\,
  p_{X_e|\bY_{-e}}(x_e|\bY_{-e})}{\sum_{x'_e}p_e(x_e'|Y_e)\,
  p_{X_e|\bY_{-e}}(x_e'|\bY_{-e})}\, .
\end{align}
Rewriting this identity in terms of $\hx_e(\bY)$, $\hx_e(\bY_{-e})$, we obtain
\begin{align}
\hx_e(\bY) &= \frac{\hx_e(\bY_{-e})
  +b(Y_e)}{1+b(Y_e)\hx_e(\bY_{-e})}\, ,\label{eq:hxe}\\
b(Y_e) & = \frac{p_e(Y_e|+1)+p_e(Y_e|-1)}{p_e(Y_e|+1)-p_e(Y_e|-1)}\, .
\end{align}
Using the definition of $p_e(y_e|x_e)$, we obtain
\begin{align}
b(y_e) = \begin{cases}
\sqrt{(1-\op_n)\theta/(n\op_n)} & \mbox{ if $y_e=+1$,}\\
-\sqrt{\op_n\theta/(n(1-\op_n))} & \mbox{ if $y_e=-1$.}
\end{cases}
\end{align}
This in particular implies $|b(y_e)|\le
\sqrt{\theta/[n\op_n(1-\op_n)]}$.
From Eq.~(\ref{eq:hxe}) we therefore get (recalling $|\hx_e(\bY_{-e})|\le 1$)
\begin{align}
\big|\hx_e(\bY)-\hx_{e}(\bY_{-e})\big| = \frac{|b(Y_e)|
  (1-\hx_e(\bY_{-e})^2)}{1+b(Y_e) \hx_e(\bY_{-e})} \le |b(\bY_e)|\le
\sqrt{\frac{\theta}{n\op_n(1-\op_n)}}\, .
\end{align}

Substituting this in Eq.~(\ref{eq:DHFinal}), we get
\begin{align}
\left|\frac{1}{n}\,\frac{\d H(\bX|\bY)}{\d \theta} +
  \frac{1}{4}\binom{n}{2}^{-1}\sum_{e\in\Pair}
\Big(1-\E\{\hx_e(\bY)^2\}\Big)\right|\le
  C\left(\sqrt{\frac{\theta}{n\op_n(1-\op_n)}} \vee \frac{1}{n}\right)\, .
\end{align}
Finally we rewrite the sum over $e\in\Pair$ explicitly as sum over
$i<j$
and recall that $X_e = X_iX_j$  to get
\begin{align}
\left|\frac{1}{n}\,\frac{\d H(\bX|\bY)}{\d \theta} +
  \frac{1}{4}\binom{n}{2}^{-1}\sum_{1\le i<j\le n}
\E\big\{\big(X_iX_j-\E\{X_iX_j|\bY\}\big)^2\big\}\right|\le
  C\left(\sqrt{\frac{\theta}{n\op_n(1-\op_n)}} \vee \frac{1}{n}\right)\, .
\end{align}
Since $\bY$ is equivalent to $\bY$ (up to a change of variables) and 
$I(\bX;\bG) = H(\bX)-H(\bG|\bG)$, with $H(\bX) = n\log 2$ is
independent of $\theta$, this is equivalent to our claim (recall the
definition of $\MMSE_n(\,\cdot\,)$, Eq.~(\ref{eq:MMSE_2vert})).
\end{proof}

\subsection{Proof of Theorem \ref{thm:MainEstimation}}

From Lemma \ref{lemma:DifferentiationSBM} and Theorem \ref{thm:main},
we obtain, for any $0<\lambda_1<\lambda_2$,
\begin{align}
\lim_{n\to\infty}\int_{\lambda_1}^{\lambda_2}\frac{1}{4}\MMSE_n(\theta)\,
\d\theta = \Psi(\gamma_*(\lambda_2),\lambda_2)-
\Psi(\gamma_*(\lambda_1),\lambda_1)\, .
\end{align}
From Lemma  \ref{lem:fixedpt} and \ref{lem:PsiIntegral} 
\begin{align}
  \lim_{n\to\infty}\int_{\lambda_1}^{\lambda_2} \MMSE_n(\theta)\,
\d\theta =
\int_{\lambda_1}^{\lambda_2}\left(1-\frac{\gamma_*(\theta)^2}{\theta^2}\right)\,\d\theta\,
.
\end{align}

\section*{Acknowledgments}

Y.D. and A.M. were partially supported by NSF grants CCF-1319979 and DMS-1106627 and the AFOSR
grant FA9550-13-1-0036. Part of this work was done while the authors were visiting
Simons Institute for the Theory of Computing, UC Berkeley.

\appendix

\section{Estimation metrics: proofs}

\subsection{Proof of Lemma \ref{lemma:VectorVsMatrix}}
\label{sec:VectorVsMatrix}

Let us begin with the upper bound  on $\vmse_n(\lambda)$.
By using $\hbx(\bG) = \E\{\bX|X_1 = +1,\bG\}$ in
Eq.~(\ref{eq:vmseDef}), we get 
\begin{align}
\vmse_n(\lambda) &\le \frac{1}{n}\,\E\Big\{\min_{s\in \{+1,-1\}}
\big\|\bX-s\E\{\bX|X_1=+1,\bG\}\big\|_2^2\Big\}\\
&=\frac{1}{n}\,\E\Big\{\min_{s\in \{+1,-1\}}
\big\|\bX-s\E\{\bX|X_1=+1,\bG\}\big\|_2^2\big|X_1 = +1\Big\}\\
&\le \frac{1}{n}\,\E\Big\{
\big\|\bX-s\E\{\bX|X_1=+1,\bG\}\big\|_2^2\big|X_1 = +1\Big\} \\
& \le \E\Big\{
\big\|\bX-\E\{\bX|X_1=+1,\bG\}\big\|_2^2\big|X_1 = +1\Big\}\, , 
\end{align}
where the equality on the second line follows because $(\bX,\bG)$ is
distributed as $(-\bX,\bG)$. The last inequality yields the desired 
upper bound $\vmse_n(\lambda) \le\MMSE_n(\lambda)$.

In order to prove the lower bound on $\vmse_n(\lambda)$ assume, for
the sake of simplicity, that the infimum in the definition
(\ref{eq:vmseDef}) is achieved at a certain estimator
$\hbx(\,\cdot\,)$. If this is not the case, the argument below can be
easily adapted by letting $\hbx(\,\cdot\, )$ be an estimator that
achieves error within $\eps$ of the infimum.

Under this assumption, we have, from (\ref{eq:vmseDef}),
\begin{align}
\vmse_n(\lambda) & =\E\min_{s\in\{+1,-1\}}\Big\{1 -
\frac{2s}{n}\, \<\bX,\hbx(\bG)\> +\frac{s^2}{n} \|\hbx(\bG)\|_2^2\Big\} \\
&\ge \E\min_{\alpha\in\reals}\Big\{1 -
\frac{2\alpha}{n}\, \<\bX,\hbx(\bG)\>+\frac{\alpha^2}{n}
  \|\hbx(\bG)\|_2^2\Big\} \\
& = 1-\E\Big\{\frac{\<\bX,\hbx(\bG)\>^2}{n\|\hbx(\bG)\|_2^2}\Big\}\, ,\label{eq:BoundVmse1}
\end{align}
where the last identity follows since the minimum over 
$\alpha$ is achieved at $\alpha = \<\bX,\hbx(\bG)\>/\|\hbx(\bG)\|_2^2$. 

Consider next the matrix minimum mean square error. 
Let $\hbx(\bG)= (\hx(\bG))_{i\in [n]}$ an optimal estimator with
respect to  $\vmse_n(\lambda)$, and define
\begin{align}
\hx_{ij}(\bG) = \beta(\bG)\, 
\hx_i(\bG)\hx_j(\bG)\, ,\;\;\;\;\;\;\;\;
\beta(\bG) \equiv \frac{1}{\|\hbx(\bG)\|_2^2}\,
  \E\left(\{\frac{\<\bX,\hbx(\bG)\>^2}{\|\hbx(\bG)\|_2^2}\right)\, .
\end{align}
Using Eq.~(\ref{eq:MMSEdef}) and the optimality of posterior mean, we obtain
\begin{align}
(1-n^{-1})\MMSE_n(\lambda)& \le
\frac{1}{n^2}\sum_{i,j\in[n]}\E\Big\{\big[X_iX_j-\hx_{ij}(\bG)\big]^2\Big\}\\
&=
\frac{1}{n^2}\E\Big\{\big\|\bX\bX^{\sT}-\beta(\bG)\hbx(G)\hbx^{\sT}(\bG)\big\|_F^2\Big\}\\
& = \E\Big\{1-
  \frac{2\beta(\bG)}{n^2}\,\<\bX,\hbx(\bG)\>^2+\frac{\beta(\bG)^2}{n^2}\,\|\hbx(\bG)\|_2^4\Big\}\\
& = 1
  -\E\left(\frac{\<\bX,\hbx(\bG)\>^2}{n\|\hbx(\bG)\|_2^2}\right)^2\, . \label{eq:BoundVmse2}
\end{align}
The desired lower bound in Eq.~(\ref{eq:VmseBounds}) follows by
comparing Eqs.~(\ref{eq:BoundVmse1}) and (\ref{eq:BoundVmse2}).

\subsection{Proof of Lemma \ref{lemma:OverlapVsVector}}
\label{sec:OverlapVsVector}

We shall assume, for the sake of simplicity, that the infimum in the
definition of  $\vmse_n(\lambda)$, see Eq.~(\ref{eq:vmseDef}) is
achieved for a given estimator $\hbx:\cG_n\to\reals^n$. If this is not
the case, the proof below is easily adapted by considering an
approximately optimal estimator. We then define
$\bxi:\cG_n\to\reals_n$ by letting
\begin{align}
\bxi(\bG) \equiv \frac{\hbx(\bG)\, \sqrt{n}}{\|\hbx(\bG)\|_2}\, .
\end{align}
Notice that $\|\bxi(\bG)\|_2 = \sqrt{n}$. Also by the proof in previous
section, see Eq.~(\ref{eq:BoundVmse1}), we have
\begin{align}
\vmse_n(\lambda)& \ge
                   1-\E\Big\{\frac{1}{n^2}\<\bX,\bxi(\bG)\>^2\Big\}\, ,
\end{align}
and therefore  (since $|\<\bX,\bxi(\bG)\>|\le n$)
\begin{align}
\E\Big\{\Big|\frac{1}{n}\<\bX,\bxi(\bG)\>\Big|\Big\} \ge
  1-\vmse_n(\lambda)\, .
\end{align}

Next consider the definition of overlap
(\ref{eq:OverlapDef}). Consider the randomized estimator $\hbs:\cG_n\to
\{+1,-1\}^n$  defined by letting $\hbs(\bG) = (\hs_i(\bG))_{i\in [n]}$
with 
\begin{align}
\hs_i(\bG) = \begin{cases}
+1 & \mbox{ with probability $(1+\xi_i(\bG))/2$,}\\
-1 & \mbox{ with probability $(1-\xi_i(\bG))/2$.}
\end{cases}
\end{align}
independently across $i\in [n]$. (Formally, $\hbs:\cG_n\times \Omega\to
\{+1,-1\}^n$ with $\Omega$ a probability space, but  we prefer to
avoid unnecessary technicalities.)

We then have, by central limit theorem
\begin{align}
\E\Big\{\frac{1}{n}\big|\<\bX,\hbs(\bG)\>\big|\; \Big|\bX,\bG\Big\} = 
\frac{1}{n}\big|\<\bX,\bxi(\bG)\>\big| +O(n^{-1/2})\,,
\end{align}
with the $O(n^{-1/2})$ uniform in $\bX,\bG$. 
This yields the desired lower bound since, by dominated convergence,
\begin{align}
\overlap_n(\lambda) &\ge
                      \E\Big\{\frac{1}{n}\big|\<\bX,\hbs(\bG)\>\big|\Big\}\\
& \ge  \E\Big\{\frac{1}{n}\big|\<\bX,\bxi(\bG)\>\big|\Big\}
  -O(n^{-1/2})\\
&\ge 1-\vmse_n(\lambda)-O(n^{-1/2})\, .
\end{align}
%
%
%*******************************************
%

%
%******************************************************
%

\section{Additional technical proofs}
\label{app:addproofs}

   \subsection{Proof of Remark \ref{rem:berns}}
    We prove the claim for $\<\bx, \bG\bx\>$; the other claim follows from an identical
    argument.
    Since $\Delta_n\le \op_n$, we have by triangle inequality, that
    $\abs{\E\{\<\bx, \bG\bx\>|X\}} \le 2n(n-1)\op_n$. Applying Bernstein
    inequality to the sum $ \<\bx, \bG\bx\> - \E\{\<\bx, \bG\bx\>|\bX\} = \sum_{i, j} x_ix_j (G_{ij} - \E\{G_{ij}|X\})$
    of random variables bounded by 1:
    \begin{align}
      \P\left\{\sup_{\bx\in\hcube} \<\bx, \bG\bx\> - \E\{\<\bx,\bG\bx\> |\bX\}\> \ge t \right\} 
      &\le  
      2^n\sup_{\bx\in \hcube}\P\left\{ \<\bx, \bG\bx\> - \E\{\<\bx,\bG\bx\> |\bX\}\> \ge t \right\} \\
     &\le 2^n \exp(-t^2/2(n^2\op_n + t))
    \end{align}
    Setting $t = Cn^2\op_n$ for large enough $C$ yields the required result.

\subsection{Proof of Lemma \ref{rem:unique}}
\label{sec:ProofUnique}

Let us start from point $(a)$,.
Since  $\mmse(\gamma, \eps) = \eps+(1-\eps)(1-\mmse(\gamma))$,
it is sufficient to prove this claim for 
  \begin{align}
    G(\gamma ) \equiv 1-\mmse(\gamma) &= \E\{ \tanh(\gamma +
    \sqrt{\gamma}Z)^2 \}\, ,
\end{align}
 where, for the rest of the proof, we keep $Z\sim\normal(0, 1)$. 
We start by noting that, for all $k\in\integers_{>0}$,
\begin{align}
\E\left\{ \tanh(\gamma + \sqrt{\gamma}Z)^{2k-1} \right\} &= \E\left\{
  \tanh(\gamma+\sqrt{\gamma}Z)^{2k} \right\},\label{eq:poweridentity}
  \end{align}
This identity can be proved using the fact  that $\E\{X |
\sqrt{\gamma} X + Z\} = \tanh(\gamma X + \sqrt{\gamma}Z)$.
Indeed this yields
\begin{align}
\E\left\{
  \tanh(\gamma+\sqrt{\gamma}Z)^{2k} \right\}& = \E\left\{
  \tanh(\gamma X+\sqrt{\gamma}Z)^{2k} \right\}\\
& = \E\left\{\E\{X |
\sqrt{\gamma} X + Z\}\,  \tanh(\gamma X+\sqrt{\gamma}Z)^{2k-1}
\right\}\\
& = \E\left\{X \,  \tanh(\gamma X+\sqrt{\gamma}Z)^{2k-1}
\right\}\\
& = \E\left\{ \tanh(\gamma X+\sqrt{\gamma}Z)^{2k-1}
\right\}\, ,
\end{align}
where the first and last equalities follow by symmetry.

Differentiating with respect to $\gamma$ (which can be justified by
dominated convergence):
  \begin{align}
    G'(\gamma) &= \E\left\{ (1+ Z/2\sqrt{\gamma}) \sech(\gamma + \sqrt{\gamma}Z)^2 \right\} \\
    &= \E\left\{ \sech(\gamma +\sqrt{\gamma}Z)^2  \right\} + 
    \frac{1}{2\sqrt{\gamma}}\E\left\{ Z\sech(\gamma+ \sqrt{\gamma}Z)^2
    \right\}.
  \end{align}
  Now applying Stein's lemma (or Gaussian integration by parts):
  \begin{align}
    G'(\gamma) &= \E\left\{ \sech(\gamma+\sqrt{\gamma}Z)^2 \right\} -
    \E\left\{ -\tanh(\gamma + \sqrt{\gamma}Z)\sech(\gamma +
      \sqrt{\gamma }Z)^2 \right\}\, .
  \end{align}
  Using the trigonometric identity $\sech(z)^2 = 1-\tanh(z)^2$,  
  the shorthand $T =\tanh(\gamma + \sqrt{\gamma}Z)$ and identity
  \eqref{eq:poweridentity} above:
  \begin{align}
    G'(\gamma) &= \E\left\{ 1-T^2 - T + T^3 \right\}\\
    &= \E\left\{ 1 - 2T^2 + T^4 \right\} = \E\left\{ (1-T^2)^2 \right\} \\
    &= \E\left\{ \sech(\gamma + \sqrt{\gamma}Z)^4 \right\}.
  \end{align}
  Now, let $\psi(z) = \sech^4(z)$, whereby we have
  \begin{align}
    G'(\gamma) &= \E\left\{ \psi(\sqrt{\gamma}(\sqrt{\gamma} + Z)) \right\}.
  \end{align}

  Note now that $\psi(z)$ satisfies $(i)$ $\psi(z)$ is even with $z\psi'(z)\le 0$, $(ii)$ $\psi(z)$
  is continuously differentiable and $(iii)$ $\psi(z)$ and $\psi'(z)$ are  bounded. 
  Consider the function $H(x, y) = \E\left\{ \psi(x(Z+y)) \right\}$, where
  $x\ge 0$. We have the identities:
  \begin{align}
    G'(\gamma) &= H(\sqrt{\gamma}, \sqrt{\gamma}),  \\
    G''(\gamma) &= \diffp{H}{x} \bigg\rvert_{x=y=\sqrt{\gamma}} + \diffp{H}{y}\bigg\rvert_{x=y=\sqrt{\gamma}}.
  \end{align}
 Hence, to prove that $G''(\gamma)$ is concave on $\gamma\ge 0$, it suffices to
 show that $\partial H/\partial x$, $\partial H/ \partial y$ are non-positive for $x, y\ge 0$.
  By properties $(ii)$ and $(iii)$ above we can differentiate
   $H$ with respect to $x, y$ and interchange differentiation and expectation. 
   
   We first prove that $\partial H/\partial x$ is non-positive:
   \begin{align}
     \diffp{H}{x} &= \E\left\{ (Z+y)\psi'(x(Z+y)) \right\} \\
     &= \int_{-\infty}^\infty \varphi(z) (z+y) \psi'(x(z+y))\d z\\
     &= \int_{-\infty}^{\infty} z\varphi(z-y) \psi'(xz)\d z. 
   \end{align}
   Here $\varphi(z)$ is the Gaussian density $\varphi(z) = \exp(-z^2/2)/\sqrt{2\pi}$.
   Since $z\psi'(z)\le 0$ by property $(i)$ and $\varphi(z-y) >0$ we have the
   required claim.

   Computing the derivative with respect to $y$ yields
   \begin{align}
     \diffp{H}{y} &= x\E\left\{ \psi'(x(Z+y)) \right\} \\
     &= \int_{-\infty}^{\infty} x\psi'(x(z +y))\varphi(z)\d z\\
     &= \int_{-\infty}^{\infty} x\psi'(xz)\varphi(z-y)\d z \\
     &= \frac{1}{2} \int_{-\infty}^{\infty} x\psi'(xz)\varphi(z-y)\d z 
     - \frac{1}{2}\int_{-\infty}^{\infty} x\psi'(xz)\varphi(y+z)\d z,
   \end{align}
   where the last line follows from the fact that $\psi'(u)$ is
   odd and $\varphi(u)$ is even in $u$. Consequently
   \begin{align}
     \diffp{H}{y} &= x\int_{0}^{\infty} \psi'(xz)  (\varphi(y-z) - \varphi(y+z)) \d z.
   \end{align}
   Since $\varphi(y-z) \ge \varphi(y+z)$ and $\psi'(xz)\le 0$ for $y,z\ge 0$, the integrand is negative and 
   we obtain the desired result.

\subsection{Proof of Remark \ref{rem:limitentro}}
\label{sec:limEntro}

For any random variable $\bR$ we have
\begin{align}
H(\bX|\bR) - H(\bX\bX^{\sT}|\bR) = H(\bX|\bX\bX^{\sT}, \bR) -
H(\bX\bX^{\sT}|\bX,\bR) = H(\bX|\bX\bX^{\sT}, \bR)\, .
\end{align}
Since $H(\bX|\bX\bX^{\sT}, \bR)\le H(\bX|\bX\bX^{\sT}) \le \log 2$
(given $\bX\bX^{\sT}$ there are exactly $2$ possible choices for
$\bX$), this implies 
\begin{align}
0 \le H(\bX|\bR) - H(\bX\bX^{\sT}|\bR) \le \log 2 \, .
\end{align}
The claim (\ref{eq:IXX}) follows by applying the last inequality once
to $\bR=\emptyset$ and once to $\bR = (\bX(\eps),\bY(\lambda))$ and
taking the difference.

The  claim (\ref{eq:IXX_lambda0})  follows from the fact that $\bY(0)$
is independent of $\bX$, and hence $I(\bX;  \bX(\eps),\bY(0)) = I(\bX;
\bX(\eps)) = n\eps\log 2$.

  For the second claim, we prove that 
  $\limsup_{n\to\infty}H(\bX|\bY(\lambda), \bX(\eps)) \le
  \delta(\lambda)$ where $\delta(\lambda)\to 0$ as $\lambda\to
  \infty$,
whence the claim follows since $H(\bX) = n\log 2$.
  We claim that we can construct an estimator $\hbx(\bY) \in \{-1, 1\}^n$ 
  and a function $\delta_1(\lambda)$ with
  $\lim_{\lambda\to\infty}\delta_1(\lambda) = 0$, such that, defining
\begin{align}
E(\lambda,n) =  \begin{cases}
1 & \mbox{ if $\min\big(d(\hbx(\bY),\bX),d(\hbx(\bY),-\bX)\big)\ge
n\delta_1(\lambda)$,}\\
0 & \mbox{ otherwise,}
\end{cases}
\end{align}
then we have
\begin{align}
\lim_{n\to\infty}\prob\big\{E(\lambda,n)=1\big\} = 0\, .
\end{align}
To prove this claim, it is sufficient to consider $\hbx(\bY)=\sign(\bv_1(\bY))$ where $\bv_1(\bY)$
is the principal eigenvector of $\bY$.
Then \cite{capitaine2009largest,benaych2011eigenvalues} implies that, for
$\lambda\ge 1$, almost surely,
\begin{align}
\lim_{n\to\infty} \frac{1}{\sqrt{n}} |\<\bv_1(\bY),\bX\>| =
\sqrt{1-\lambda^{-1}} \, .
\end{align}
Hence the above claim holds, for instance, with $\delta_1(\lambda)
= 2/\lambda$.

  Then expanding $H(\bX, E|\bY(\lambda), \bX(\eps))$ with the chain
  rule (whereby $E=E(\lambda,n)$), we get:
  \begin{align}
    H(\bX|\bY(\lambda),\bX(\eps)) + H(E|\bX, \bY(\lambda), \bX(\eps)) &= H(\bX, E|\bY(\lambda), \bX(\eps))\\
    & = H(E|\bY(\lambda), \bX(\eps)) + H(\bX|E, \bY(\lambda),\nonumber
    \bX(\eps)).
  \end{align}
  Since $E$ is a function of $\bX, \bY(\lambda)$, $H(E|\bX, \bY(\lambda)) = 0$. 
  Furthermore $H(E|\bY(\lambda),\bX(\eps)) \le \log 2$
  since $E$ is binary. Hence:
  \begin{align}
    H(\bX|\bY(\lambda), \bX(\eps)) &\le\log 2 + H(\bX|E, \bY(\lambda), \bX(\eps)) \\
    &= \log 2 + \P\{E=0\}H(\bX|E = 0, \bY(\lambda), \bX(\eps)) + \P\{E=1\} H(\bX| E=1, \bY(\lambda), \bX(\eps)).\nonumber
  \end{align}
  When $E = 0$, $\bX$ differs from $\pm\hbx(\bY)$ in at most $\delta_1 n$ positions, whence
  $H(\bX|E=0, \bY(\lambda), \bX(\eps)) \le n \delta_1\log (e/\delta_1)
  +\log 2$. 
  When $E=1$, we trivially have $H(\bX|E=1, \bY(\lambda), \bX(\eps)) \le H(\bX) = n$.
  Consequently:
  \begin{align}
    H(\bX|\bY(\lambda), \bX(\eps)) &\le 2\log 2 + n\delta_1 \log \frac{e}{\delta_1} + n \delta_1. 
  \end{align}
  The second claim then follows by dividing with $n$ and letting $n\to\infty$ on the right
  hand side.

\subsection{Proof of Lemma \ref{lem:stateevollem}}

The lemma results by reducing the AMP algorithm 
Eq. \eqref{eq:ampiter} to the setting of \cite{javanmard2013state}.

By definition, we have:
\begin{align}
  \bx^{t+1} &= \frac{\bY(\lambda)}{\sqrt{n}} f_t(\bx^t, \bX(\eps)) - \ons_tf_{t-1}(\bx^{t-1}, \bX(\eps)) \\
  &= \frac{\sqrt{\lambda}}{n} \<\bX, f_t(\bx^t, \bX(\eps))\>\bX + \bZ
  f_t(\bx^t, \bX(\eps)) - \ons_t f_{t-1}(\bx^{t-1}, \bX(\eps)).
\end{align}
Define a related sequence $\bs^t\in\reals^n$ as follows:
\begin{align}
  \bs^{t+1}  &= \bZ f_t(\bs^t + \mu_t \bX, \bX(\eps)) - \tons_t f_{t-1}(\bs^{t-1} + \mu_{t-1}\bX, \bX(\eps)) \label{eq:ampproofiter}\\
  \tons_t &= \frac{1}{n}\sum_{i\in[n]} f'_t(s_i^t + \mu_t X_i,
  X(\eps)_i)\, , \\
  \bs^0 &= \bx^0 + \mu_0\bX\, .
\end{align}
Here $\mu_t$ is defined via the state evolution recursion:
\begin{align}
R
  \mu_t &= 
  \sqrt{\lambda}\, \E\left\{ X_0 f_t(\mu_t X_0 + \sigma_t Z_0 ,
    X_0(\eps) )\right\}\, ,\\
  \sigma^2_t &= \E\left\{ f_t(\mu_t X_0 + \sigma_t Z_0,  X_0(\eps) )^2
  \right\}, .
\end{align}
We call a function $\psi:\reals^d\to\reals$ is pseudo-Lipschitz if,
for all $\bu,\bv\in\reals^d$
\begin{align}
  \abs{\psi(\bu)-\psi(\bv)} &\le L(1+\|\bu\|+\|\bv\|)\,\|\bu-\bv\|\, ,
\end{align}
where $L$ is a constant. In the rest of the proof, we will use $L$ to
denote a constant that may depend on $t$ and but not on $n$, and can
change from line to line.

We are now ready to prove Lemma \ref{lem:stateevollem}. Since the iteration for $\bs^t$ is in 
the form of \cite{javanmard2013state}, we have
for any pseudo-Lipschitz function $\tilde{\psi}$:
\begin{align}
  \lim_{n\to\infty} \frac{1}{n} \sum_{i=1}^n\tilde{\psi}(s^t_i, X_i,X(\eps)_i) &=
  \E\big\{ \tilde{\psi}(\sigma_t Z_0, X_0, X_0(\eps)=q
  )\big\}\, . \label{eq:jmthm}
\end{align}
Letting $\tilde{\psi}(s, z,r) = \psi(s + \mu_t z, z,r)$, this  implies that,
almost surely:
\begin{align}
  \lim_{n\to\infty} \frac{1}{n} \sum_{i=1}^n\psi(s^t_i+ \mu_t X_i,
  X_i,X(\eps)_i) &= \E\big\{ \psi(\mu_t X_0+ \sigma_t Z_0, X_0,
  X_0(\eps))\big\}\, .
\end{align}
It then suffices to show that, for any pseudo-Lipschitz function
$\psi$,  almost surely:
\begin{align}
  \lim_{n\to\infty} \frac{1}{n}\sum_{i=1}^n \big[\psi(s^t_i + \mu_t
  X_i, X(\eps)_i) - \psi(x^t_i, X_i,X(\eps)_i) \big] = 0\, .
\end{align}

We instead prove the following claims that include
the above. For any $t$ fixed, almost surely:
\begin{align}
  \lim_{n\to\infty} \frac{1}{n}\sum_{i=1}^n
\big[\psi(s^t_i + \mu_t X_i, X(\eps)_i) - \psi(x^t_i, X_i, X(\eps)_i)\big] &= 0,\label{eq:induc0}\\ 
  \lim_{n\to \infty}\frac{1}{n}\,\|\bDelta^t\|_2^2  &= 0, \label{eq:induc1}\\
  \limsup_{n\to\infty} \frac{1}{n}\,\|\bs^t +
  \mu_t \bX\|_2^2 &<\infty\, ,\label{eq:induc2}
\end{align}
where  we let $\bDelta^t = \bx^t - \bs^t - \mu_t \bX$. 

We can prove this claim by induction on $t$. 
The base case of $t=-1, 0$ is trivial for all three claims: 
$\bs^0 +\mu^0 \bX = \bx^0$ and $\bDelta^0 = 0$ is satisfied by our initial condition 
$\bs^0 = \bx^0 = 0$, $\mu_0 = 0$. Now, assuming the claim
holds for $\ell=0, 1, \dots t-1$ we prove the claim for $t$.

By the pseudo-Lipschitz property and triangle inequality, we have, for
some $L$:
\begin{align}
  \abs{\psi(x^t_i, X_i,X(\eps)_i) - \psi(s^t_i + \mu_t X_i ,
    X_i,X(\eps)_{i})}
&\le L \abs{\Delta^t_i} (1+ \abs{s^t_i + \mu_t X_i} +
\abs{x^t_i})\\
  &\le 2L(\abs{\Delta^t_i} + \abs{s^t_i + \mu_t X_i} \abs{\Delta^t_i} +
\abs{\Delta^t_i}^2).
\end{align}
Consequently:
\begin{align}
  \frac{1}{n}\abs{\sum_{i=1}^n\big[ \psi(x^t, X_i,X(\eps)_i) - \psi(s^t+\mu_tX_i, X_i,X(\eps)_i)\big]} &\le
  \frac{L}{n} \sum_{i=1}^n\left( \abs{\Delta^t_i} + 
\abs{s^t_i + \mu_t X_i} \abs{\Delta^t_i} + \abs{\Delta^t_i}^2\right)\\
&\le \frac{L}{n} \Big(\|\bDelta^t\|_2^2 +\sqrt{n}\|\bDelta^t\|_2+
\|\bDelta^t\|_2\|\bs^t+\mu_t\bX\|_2\Big)\, .
\end{align}
Hence the induction claim \myeqref{eq:induc0} at $t$ follows from claims \myeqref{eq:induc1} and
\myeqref{eq:induc2} at $t$.

We next consider the claim \myeqref{eq:induc1}.
Expanding the iterations for $\bx^t, \bs^t$ we obtain the following 
expression for $\Delta^t_i$:
\begin{align}
  \Delta^t_i &= \left( \frac{\sqrt{\lambda}\<f_{t-1}(\bx^{t-1}, \bX(\eps)), \bX\>}{n} - \mu_t\right) X_i +
  \frac{1}{\sqrt{n}}\<\bZ_i, f_{t-1}(\bx^{t-1}, \bX(\eps)) - f_{t-1}(\bs^{t-1} + \mu_{t-1} \bX, \bX(\eps))\> 
  \nonumber\\
  &\quad- \ons_{t-1} f_{t-2}(x^{t-2}_i, X(\eps)_i) + \tons_{t-1} f_{t-2}(s^{t-2}_i + \mu_{t-2}X_i, X(\eps)_i). 
 % x^{t+1} - s^{t+1}-\mu_t X &= \left( \frac{\sqrt{\lambda}\<f_t(x^t, X(\eps)), X\>}{n} - \mu_{t+1} \right) X
 % + Z(f_t(x^t, X(\eps))  -f_t(s^t + \mu_t X, X(\eps))) - 
\end{align}
Here $\bZ_i$ is the $i^{\text{th}}$ row of $\bZ$.

Now, with the standard inequality $(z_1 + z_2 + z_3)^2 \le 3(z_1^2 +
z_2^2 + z_3^2)$:
\begin{align}
  \frac{1}{n}\|\bDelta^t\|_2^2 &\le 
  L\bigg(\frac{\sqrt{\lambda}}{n}\<\bX, f_{t-1}(\bx^{t-1}, \bX(\eps))\> - \mu_t\bigg)^2 \nonumber\\
  &\quad+\frac{L}{n^2} \|\bZ\|_2^2\norm{f_{t-1}(\bx^{t-1}, \bX(\eps))
    - f_{t-1}(\bs^{t-1} + \mu_{t-1}\bX, \bX(\eps))}^2
 \nonumber \\
  &\quad+ L \abs{\tons_{t-1} - \ons_{t-1}}^2 \frac{1}{n}\sum_{i=1}^nf_{t-2}(s^{t-2}_i+\mu_{t-2}X_i, X(\eps)_i ) \nonumber\\
  &\quad+ L\abs{\ons_{t-1}}^2 \frac{1}{n} \sum_{i=1}^n (f_{t-2}(s^{t-2}_i + \mu_{t-2}X_i, X(\eps)_i) - 
  f_{t-2}(x^{t-2}_i, X(\eps)_i))^2. 
\end{align}
Using the fact that $f_{t-1}, f_{t-2}$ are Lipschitz:
\begin{align}
  \frac{1}{n}\|\bDelta^t\|_2^2 &\le 
  L\bigg(\frac{\sqrt{\lambda}}{n}\<\bX, f_{t-1}(\bx^{t-1},
  \bX(\eps))\> - \mu_t\bigg)^2 + 
\frac{L}{n^2}\|\bZ\|_2^2\norm{\bDelta^{t-1}}^2_2 \nonumber \\
  &\quad+ L \abs{\tons_{t-1} - \ons_{t-1}}^2 \frac{1}{n}\sum_{i=1}^n f_{t-2}(s^{t-2}_i+\mu_{t-2}X_i, X(\eps)_i ) \nonumber\\
  &\quad+ L\abs{\ons_{t-1}}^2 \frac{\norm{\Delta^{t-2}}^2}{n}. \label{eq:deltaclaim}
\end{align}
By the induction hypothesis, (specifically $\psi(x, y,r)= yf_{t-1}(x)$ at $t-1$, wherein it
is immediate to check that $yf_{t-1}(x)$ is pseudo-Lipschitz by the boundedness of $\mu_t, \sigma_t$):
\begin{align}
  \lim_{n\to\infty} \frac{\<\bX, f_{t-1}(\bx^{t-1}, \bX(\eps))\>}{n} & = 
  \lim_{n\to\infty} \frac{\<\bX, f_{t-1}(\bs^{t-1} + \mu_{t-1}\bX, \bX(\eps))\>}{n} 
  = \mu_t \quad\text{a.s.}
\end{align}
Thus the first term in \myeqref{eq:deltaclaim} vanishes.
For the second term to vanish, using the induction hypothesis
for $\Delta^{t-1}$, it suffices that almost surely:
\begin{align}
  \limsup_{n\to\infty} \frac{1}{n}\norm{\bZ}^2&<\infty.
\end{align}
This follows from standard eigenvalue bounds for Wigner random matrices
\cite{Guionnet}.
For the third term in \myeqref{eq:deltaclaim} to vanish, we have by \cite{javanmard2013state}
that:
\begin{align}
  \limsup_{n\to\infty} \frac{1}{n}\sum_{i=1}^n f_{t-2}(s^{t-2}_i + \mu_{t-2}X_i, X(\eps)_i) <\infty. 
\end{align}
Hence it suffices that $\lim_{n\to\infty}\tons_{t-1} - \ons_{t-1} =0$ a.s., for which we expand their
definitions to get:
\begin{align}
  \lim_{n\to\infty}\tons_{t-1}-\ons_{t-1} &= \lim_{n\to\infty}
  \frac{1}{n}\sum_{i=1}^n\big[ f'_{t-1}(s^{t-1}_i+ \mu_{t-1}X_i, X(\eps)_i) - f'_{t-1}(x^{t-1}_i, X(\eps)_i)].
\end{align}
By assumption, $f_{t-1}'$ is Lipschitz and we can apply the induction hypothesis
with $\psi(x, y,q) = f'_{t-1}(x, q)$ 
to obtain that the limit vanishes. Indeed, by a similar argument $\tons_{t-1}$ is
bounded asymptotically in $n$, and so is $\ons_{t-1}$. Along with the induction
hypothesis for $\bDelta^{t-2}$ this implies that the fourth term in \myeqref{eq:deltaclaim}
asymptotically vanishes. This establishes the induction claim \myeqref{eq:induc1}. 

Now we only need to show the induction claim \myeqref{eq:induc2}. However, this
is a direct result of Theorem 1 of \cite{javanmard2013state}: indeed in \myeqref{eq:jmthm}
we let $\tilde{\psi}(x, y) = (s^t_i + \mu_t X_i)^2$ to obtain the required
claim.

\bibliographystyle{amsalpha}

\providecommand{\bysame}{\leavevmode\hbox to3em{\hrulefill}\thinspace}
\providecommand{\MR}{\relax\ifhmode\unskip\space\fi MR }
\providecommand{\MRhref}[2]{%
  \href{http://www.ams.org/mathscinet-getitem?mr=#1}{#2}
}
\providecommand{\href}[2]{#2}

\end{document}